\documentclass[lettersize,journal]{IEEEtran}
\usepackage{amsmath,amsfonts,amsfonts, amssymb, graphicx,amsthm}
\usepackage{algorithm}
\usepackage{algpseudocode}
\usepackage{array}
\usepackage{textcomp}
\usepackage{stfloats}
\usepackage{url}
\usepackage{verbatim}
\usepackage{cite}
\usepackage{dsfont}
\usepackage{bm}  
\usepackage{mathtools}
\usepackage{caption}
\usepackage{subcaption}
\usepackage{xcolor}

\newcommand{\bxi}{\bm{\xi}}
\newcommand{\bpsi}{\bm{\psi}}
\newcommand{\bmu}{\bm{\mu}}
\newcommand{\bmeta}{\bm{\eta}}
\newcommand{\btheta}{\bm{\theta}}
\newcommand{\e}{\mathbb{E}}
\newcommand{\f}{\mathcal{F}}
\newcommand{\bmf}{\boldsymbol{\mathcal{F}}}
\newcommand{\w}{\bm{w}}
\newcommand{\m}{\bm{m}}

\newcommand{\bvart}{\bm{\upsilon}}
\newcommand{\dkl}{D_{\textup{KL}}}

\newcommand{\mcl}{\text{col}}
\newcommand{\bT}{\mathbb{T}}
\DeclareMathOperator{\T}{\mathsf{T}}
\DeclareMathOperator*{\argmax}{arg\,max}

\DeclareMathOperator{\rank}{rank}
\DeclareMathOperator{\range}{RAN}
\DeclareMathOperator{\mynull}{NULL}

\makeatletter
\def\hlinewd#1{%
	\noalign{\ifnum0=`}\fi\hrule \@height #1 \futurelet
	\reserved@a\@xhline}
\makeatother
\makeatletter
\renewcommand{\fnum@algorithm}{\fname@algorithm}
\makeatother

 \algdef{SE}[FORA]{FORA}{ENDFORA}[1]{\algorithmicfor\ #1\ \textbf{locally}}{\algorithmicend\ \algorithmicfor}
  \algdef{SE}[FORB]{FORB}{ENDFORB}[1]{\algorithmicfor\ #1\ \textbf{evolve}}{\algorithmicend\ \algorithmicfor}
   \algdef{SE}[FORC]{FORC}{ENDFORC}[1]{\algorithmicfor\ #1\ \textbf{adapt}}{\algorithmicend\ \algorithmicfor}
    \algdef{SE}[FORD]{FORD}{ENDFORD}[1]{\algorithmicfor\ #1\ \textbf{combine}}{\algorithmicend\ \algorithmicfor}

\newtheorem{theorem}{Theorem}

\newtheorem{lemma}{Lemma}
\newtheorem{corollary}{Corollary}
\newtheorem{assumption}{Assumption}
\newtheorem{remark}{Remark}

\begin{document}

\title{Distributed Bayesian Learning of Dynamic States}

\author{Mert Kayaalp, Virginia Bordignon, Stefan Vlaski, Vincenzo Matta, and Ali H. Sayed
       \thanks{M. Kayaalp, V. Bordignon and A. H. Sayed are with the Adaptive Systems Laboratory, \'{E}cole Polytechnique F\'{e}d\'{e}rale de Lausanne (EPFL), CH-1015, Switzerland. S. Vlaski is with the Department of Electrical and Electronic Engineering, Imperial College London, London SW7 2AZ, UK. V. Matta is with the Department of Information and Electrical Engineering and Applied Mathematics (DIEM), University of Salerno, via Giovanni Paolo II, I-84084, Fisciano (SA), Italy. Emails: \{mert.kayaalp, virginia.bordignon, ali.sayed\}@epfl.ch, s.vlaski@imperial.ac.uk , vmatta@unisa.it .}
        \thanks{This work was supported in part by SNSF grant 205121-184999. A short version of this work was presented in \cite{kayaalp2022dslw}.}}
 

\maketitle

\begin{abstract}
This work studies networked agents cooperating to track a dynamical state of nature under partial information. The proposed algorithm is a distributed Bayesian filtering algorithm for finite-state hidden Markov models (HMMs). It can be used for sequential state estimation tasks, as well as for modeling opinion formation over social networks under dynamic environments. We show that the disagreement with the optimal centralized solution is asymptotically bounded for the class of geometrically ergodic state transition models, which includes rapidly changing models. We also derive recursions for calculating the probability of error and establish convergence under Gaussian observation models. Simulations are provided to illustrate the theory and to compare against alternative approaches.
\end{abstract}


\section{Introduction}\label{sec:introduction}
\IEEEPARstart{D}{istributed} inference refers to the problem where a collection of agents works collaboratively in order to estimate a hidden variable of interest. This problem is of fundamental importance to the design of distributed systems, as well to the study of opinion formation over social networks. The goal in the first case is to devise communication/computation-efficient algorithms that approach centralized levels of performance by relying solely on localized agent interactions. In comparison, the goal in the second case is to reason about behavioral phenomena occurring during social learning processes. 

The hidden variable (also called state or hypothesis) that the agents are interested in tracking is time-varying in many scenarios, such as the position of a moving object, the concentration of air pollutants, and the product quality of a brand. In all situations, the agents will attempt to cooperatively track the dynamic state by using observations emitted by the underlying physical systems.

This setting is general enough and can be used in many engineering applications, including target tracking, environmental monitoring, and opinion formation over networks. For example, consider an economic network where the individual agents are trying to decide which currency (e.g., USD, EUR, CHF) is the best option to buy now. The optimal choice (true hypothesis) can be changing rapidly. Most of the literature on social learning ignores the dynamic nature of the truth, or assumes slow transition models.

In this work, we propose a networked filtering algorithm to track the state of a general hidden Markov model (HMM). We also analyze the performance and steady-state behavior of the resulting distributed strategy. In this process, we clarify questions about the benefit of cooperation and the nature of equilibria in social networks under dynamic environments. More specifically, the following is a list of the main contributions:
\begin{itemize}
    \item In Sec. \ref{sec:dif_HMM}, we propose an HMM filtering algorithm for multi-agent networks. The algorithm requires only one round of communication between agents per state change. Moreover, it utilizes the knowledge of the transition model, which allows it to track highly dynamic states. 
    \item In Sec. \ref{sec:optimality_gap}, we study the deviation of the proposed algorithm from the optimal centralized strategy defined in Sec. \ref{sec:centralized_HMM}. Geometric ergodicity is the only assumption on the transition model. The specialization of the results to the single-agent case is also another contribution to the Bayesian filtering literature. 
    \item In Sec. \ref{sec:probability_error}, we provide recursive expressions for the probability of error across the network for the binary hypothesis testing case. Furthermore, under Gaussian data distributions, we obtain an asymptotic convergence result in \emph{distribution}. The result implies that the agents attain steady-state probability of errors, which can vary across the agents depending on their centrality.
    \item In Sec. \ref{sec:simulations}, we support the theoretical results with simulations. Furthermore, we compare the proposed algorithm to alternative methods, which are described in Sec. \ref{sec:alternative_algorithms}.
\end{itemize}

\section{Related Work}

\subsection{Distributed State Estimation}

We first comment on works related to the tracking of the state evolution of dynamical models, which is one of the applications where our methodology will apply. One traditional approach to distributed state estimation includes distributed Kalman filters \cite{olfati_2007,khan_2008,cattivelli_2010,talebi2021,qian2022,moradi2022}. These are based on {\em linear} dynamical models, which we do not assume in the present work. For general transition and observation models, a common approach is the useful Bayesian methodology \cite{hlinka_2012, battistelli_2014, dedecius_2017, bandyopadhyay_2018}. Based on the observed data, agents recursively update distributions over the set of states, which are called \emph{beliefs}. This approach exploits the complete information of a distribution as opposed to some statistics of it such as the mean. Moreover, in the discrete finite state-space case (which is the focus of this work), maximizing the belief computed as the {\em exact} Bayes posterior is optimal in terms of error probability \cite{krishnamurthy_2016}. For the multi-agent case, the optimal solution requires collecting all data at a central location, which might be problematic in terms of robustness and privacy. Hence, we adopt a \emph{distributed} Bayesian point of view, and use the centralized strategy as a baseline.

The works \cite{hlinka_2012, battistelli_2014} consider distributed Bayesian state estimation; nevertheless, they require multiple rounds of communication and consensus between agents per state change. This might not be possible under highly dynamic environments. We show in the sequel that it is sufficient to communicate {\em once} per iteration for the algorithm proposed in this work in order to guarantee a bounded disagreement with the optimal centralized solution. Moreover, the works \cite{hlinka_2012, dedecius_2017} consider observation models that are restricted to the exponential family of distributions for computational tractability. In the present work, we assume that the true state can take values from a finite set. Consequently, confining to analytically well-behaved distributions only, or using sample-based filters such as particle filters \cite{ustebay11,papa2019,bordin2022}, will not be necessary. The work \cite{bandyopadhyay_2018} proposes a distributed Bayesian filtering (DBF) algorithm with one round of communication at each iteration. The agents combine their neighbors' likelihoods from previous time with the likelihood of their own fresh observation to perform the update. However, as their modeling assumptions suggest, this approach would work only if the rate of change in the likelihood of the data and in the state evolution are slow. In the present work, the proposed algorithm requires fusing the neighbors' previous beliefs, and then \emph{time-adjusting} them before combining with the likelihood of the newly arrived data. The time-adjustment exploits the transition model and allows the agents to track fast changes in the state. Accordingly, the optimality gap established in Theorem \ref{th:kl_without_net_assumption} is shown to remain bounded for a large class of transition models, including fast mixing Markov chains. In Section \ref{sec:simulations}, we support with simulation results that time-adjustment before combining with fresh information increases the performance of the proposed strategy in comparison to DBF \cite{bandyopadhyay_2018}. More recently, reference \cite{calvo2022distributed} extended \cite{bandyopadhyay_2018} by using selective information sharing in order to reduce the communication complexity. Other related works include the work \cite{bruno2018}, which provides a Bayesian interpretation for distributed Kalman filters, and the work \cite{hua2019}, which proposes a filtering algorithm for linear models with unknown covariances.

\subsection{Social Learning}

We next comment on works related to learning and tracking the state of nature over social networks, which is another application where our proposed strategy will be applicable. In the social learning context, agents form opinions (or beliefs) about the underlying state of the environment by using their private observations and by interacting with neighboring agents \cite{acemoglu_2011, krishnamurthy_2013, bordignon_2021, ping2022, jadbabaie_2012}. In \emph{locally} Bayesian social learning\footnote{These algorithms are also called non-Bayesian learning in the literature.}, agents repeatedly update their beliefs based on new observations and combine them with their neighbors' beliefs using consensus \cite{jadbabaie_2012,nedic_2017,hare2020} or diffusion \cite{zhao_2012,lalitha_2018,valentina2021,inan2022social} strategies.

All these works assume a fixed state of nature, although in many settings the state is evolving over time. For example, the works \cite{acemoglu_2008, frongillo_2011, shahrampour_2013,dasaratha2018learning} model the environment as a linear dynamical system. To address non-stationary environments where the state of nature evolves over time, reference \cite{bordignon_2021} proposes the \emph{adaptive social learning} (ASL) strategy. The algorithm incorporates a step-size parameter that infuses adaptation into its operation; the algorithm, however, does not exploit or assume any existing transition model for the state. In many applications, it happens that the current state makes some future states more likely to occur than others. For example, consider again the same economic network mentioned earlier dealing with a choice of currencies. In this application, the agents could exploit information about existing correlation between currency pairs, e.g., it is widely known that gold and the Swiss Franc are highly correlated. In Section \ref{sec:simulations}, we show by means of numerical simulations that making use of the transition model can significantly increase the tracking performance of the proposed strategy in comparison to ASL. \\

\noindent\textbf{Notation:} Boldface letters are reserved for random variables, e.g., \( \bm{x}_i \). All logarithms are natural logarithms. The symbol \( \mathds{1}_{K} \) denotes the all-ones vector of size \( K \). For a probability mass distribution \( \mu (\theta) \) over a finite set \( \theta \in \Theta \), the notation \( \mu (\theta) \propto \digamma (\theta)\) refers to the normalization:
\begin{align}
    \mu (\theta) = \frac{\digamma (\theta)}{\sum_{\theta^\prime \in \Theta }\digamma (\theta^\prime)}.
\end{align}
The notation \( \dkl(\mu_1 || \mu_2) \) represents the Kullback-Leibler (KL) divergence between the distributions \( \mu_1 \) and \( \mu_2 \).

\section{Problem Formulation}\label{sec:problem_formulation}

We consider a network of \( K \) agents, denoted by \( \mathcal{N}\), that are cooperating to track some \emph{dynamic} state of nature, denoted by \(\btheta_i^\circ\) at time \(i\). Agents exchange \emph{beliefs} with each other while respecting a graph topology restricting the communication. The belief \( \mu_{k,i} (\theta) \) at agent $k$ and time $i$ is a probability distribution over a finite set of \( H \) possible hypotheses, i.e., \(\theta\in\Theta = \{0,1,..,H-1\}\). The value \( \mu_{k,i}(\theta)\) represents the confidence level that agent \( k \) has at time \( i \) about \( \theta \) being the true hypothesis \( \btheta_i^\circ \) (which is also assumed to belong to the set \( \Theta \)). The true hypothesis is a random variable and it will be assumed to evolve according to some Markov chain, known to all agents. We use the following notation for the transition model:
\begin{align}
\bT (\theta_i | \theta_{i-1}) \triangleq \mathbb{P} (\btheta_i^\circ = \theta_i | \btheta_{i-1}^\circ = \theta_{i-1}).
\end{align}
The problem setting is as follows. At each time instant \( i \), each agent \( k \) receives a partially informative observation \( \bxi_{k,i}  \) about the true hypothesis \( \btheta_i^\circ \). Conditioned on \( \btheta_i^\circ \), the observation is distributed according to some likelihood function known to agent \( k \), and denoted by \( L_k (\bxi_{k,i} | \btheta_i^\circ)\). These agent-specific likelihoods can be probability density or mass functions depending on whether the observations are continuous or discrete. In the sequel, for ease of notation, we assume the observations are continuous. Nevertheless, our analysis is also valid for discrete observations with proper adjustments, e.g., by changing integrals to summations. Next, we state some common assumptions.

\begin{assumption}[{\bf Independent observations}\cite{hlinka_2012,shahrampour_2013}]\label{as:independence}
Conditioned on the true state, the observations are independent over space. More specifically, let \( \bxi_i \triangleq \{ \bxi_{k,i} \}_{k=1}^K \), collect all observations from across the agents at time \( i \). Then, the joint likelihood is given by, 
\begin{align}
L (\bxi_i | \btheta_i^\circ) = \prod_{k=1}^K L_k (\bxi_{k,i} | \btheta_i^\circ).
\end{align}
\end{assumption} \qed

In this work, agents will be required to communicate only once with their neighbors per iteration. The underlying communication topology is assumed to satisfy the following condition. 
\begin{assumption}[{\bf Strongly connected graph} \cite{hlinka_2012,shahrampour_2013,nedic_2017}]\label{as:network_top}
The graph topology is strongly connected \cite{sayed_2014}, which means that there exists a path between any pair of agents \( (k,\ell)\), and, moreover, there exists at least one agent \( k_{\circ}\) with a self-loop (i.e., \(a_{k_\circ,k_\circ}>0\)). Under these conditions, the combination matrix \( A = [a_{\ell k}]\) turns out to be primitive. Here, the entry \( a_{\ell k} \geq 0 \) is the weight that agent \( k \) uses to scale information sent by agent \( \ell \). This weight will be positive if, and only if, agent \(\ell\) is in the neighborhood of \(k\), written as \( \ell\in\mathcal{N}_k\)---see Figure~\ref{fig:network_problem} for a diagram representation. Moreover, we assume that $A$ is a doubly-stochastic and symmetric matrix, namely,
\begin{align}
A\mathds{1}_{K}=\mathds{1}_{K}, \quad A=A^{\T}.
\end{align}
\qed
\end{assumption} 
\noindent Strong connectivity causes the information to disperse throughout the entire network given sufficient iterations. When the true hypothesis is fixed (i.e., \( \btheta_i^\circ=\theta^\circ\)), this allows the agents to reach agreement and learn \(\theta^\circ\) almost surely \cite{jadbabaie_2012,zhao_2012,nedic_2017,lalitha_2018}. However, strong connectivity is not sufficient for network agreement if the true hypothesis is changing rapidly before local information reaches other agents, as in the current work. 

\begin{figure}[ht]
	\centering
	\includegraphics[width=2.2in]{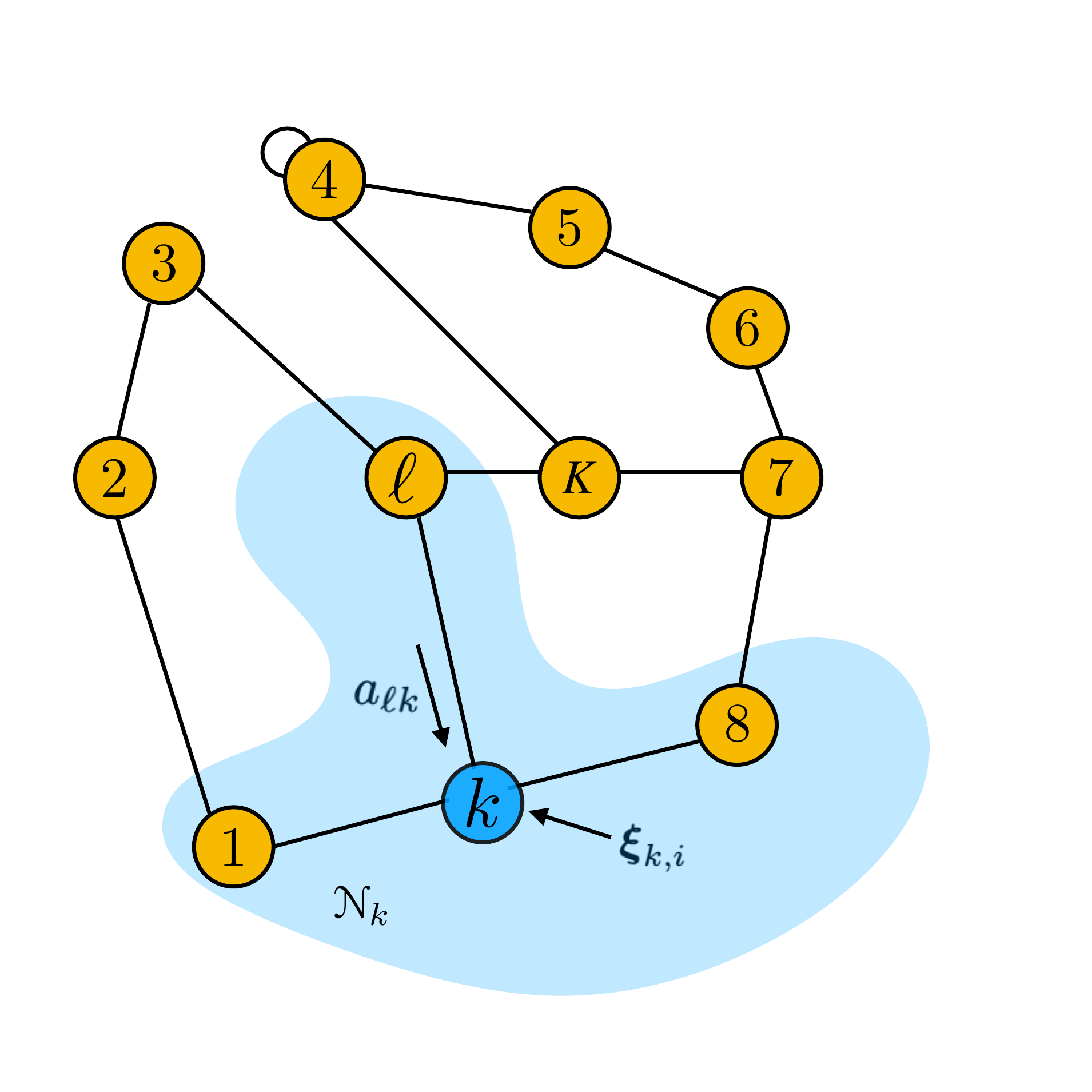}
	\caption{Agents are connected by a graph topology. The neighborhood of agent $k$ is highlighted in blue color.}
	\label{fig:network_problem}
\end{figure}

\subsection{Optimal Centralized Belief Recursion}\label{sec:centralized_HMM}
Let us denote the observation history of all agents across the network up to time \(i\) by \( \bmf_i \triangleq \{\bxi_j \}_{j=1}^i \). Likewise, let us denote the posterior distribution (or belief), which is a probability mass function (pmf) due to the assumed finite state-space model, by the notation:
\begin{align}\label{eq:true_posterior}
\bmu_i^\star (\theta_i) \triangleq \mathbb{P} (\btheta_i^\circ = \theta_i | \bmf_i).
\end{align}
It is known that the above distribution satisfies the optimal Bayesian filtering recursion \cite{krishnamurthy_2016}:
\begin{align}\label{eq:centralized_posterior}
 \bmu_i^\star (\theta_i) \propto L(\bxi_i|\theta_i) \bmeta_i^\star (\theta_i),
\end{align}
where \( \bmeta_i^\star (\theta_i) \) is the time-adjusted prior defined by
\begin{align}\label{eq:ck_cent}
\bmeta_i^\star (\theta_i) &\triangleq  \mathbb{P} (\btheta_i^\circ = \theta_i | \bmf_{i-1}) \notag \\
&=\sum_{\theta_{i-1}\in\Theta} \bT(\theta_{i}| \theta_{i-1}) \bmu_{i-1}^\star(\theta_{i-1}).
\end{align}
Once the posterior is updated by \eqref{eq:true_posterior}, the state estimator at time \( i \) is obtained from the maximum a-posteriori construction:
\begin{align}\label{eq:centralized_estimator}
    \widehat{\btheta}_{i}^{\star} \triangleq \argmax_{\theta_i\in \Theta} \bmu_i^\star (\theta_i) 
\end{align}
The main challenge with this solution method is that it requires a fusion center to gather all data from across time and agents. In Section \ref{sec:optimality_gap}, we will examine how close the beliefs generated by the proposed decentralized algorithm will get to the above centralized posterior given by \eqref{eq:centralized_posterior}--\eqref{eq:ck_cent}.
\begin{remark}[{\bf Sequence estimation}]
In this work, the focus is on estimating the \emph{current} state from past observations (i.e., causal estimation). If a sequence of hidden states, including both future or past states, is to be estimated, then single state-estimators can be combined with dynamic programming principles, such as in the Viterbi algorithm \cite{viterbi_1967,krishnamurthy_2016}.\qed
\end{remark}

\section{Decentralized Bayesian Filtering}
The centralized solution \eqref{eq:centralized_posterior}--\eqref{eq:ck_cent} can be disadvantageous for various reasons: (i) collecting all data at a single fusion location makes the system vulnerable with a single point of failure; (ii) the agents may be reluctant to share their raw data with a remote central processor for privacy or security reasons; and (iii) communications back and forth with a remote fusion center is costly. For these reasons, we pursue instead a decentralized solution that is able to approach the performance of the centralized solution. In the decentralized approach, agents will only share data with their immediate neighbors; actually, the agents will not be required to share their raw data but only their updated belief vectors. The resulting solution will be more robust to node or link failure and more communication efficient, and will lead to an effective solution method.

\subsection{Diffusion HMM Filtering}\label{sec:dif_HMM}
The streaming observations arriving at each agent are generally only partially informative about the true state of nature, \( \btheta_i^\circ\). For this reason, agents will need to cooperate with their neighbors, thus leading to a learning mechanism that allows information to diffuse through the network for enhanced performance. To do so, we propose a \emph{social} Bayesian filtering algorithm where cooperation among agents takes advantage of the notion of diffusion learning (see, e.g., \cite{sayed_2014,sayedProc}). 

Specifically, at  each time instant \( i \), every agent \( k \) first time-adjusts or evolves its belief from \( i-1 \), denoted by \(\bmu_{k,i-1}(\theta_{i-1})\),  via the Chapman-Kolmogorov equation \cite{krishnamurthy_2016} and generates an updated prior denoted by \(\bmeta_{k,i}(\theta_i)\):
\begin{align}\label{eq:dif_evolve_step_1}
\bmeta_{k,i} (\theta_{i}) &= \sum_{\theta_{i-1} \in \Theta} \bT(\theta_{i}| \theta_{i-1}) \bmu_{k,i-1}(\theta_{i-1}) \qquad \text{(Evolve)}
\end{align}
This relation is motivated by the optimal update \eqref{eq:ck_cent}, except that optimal beliefs are replaced by their local versions at agent \(k\). In the next step, agents seek to incorporate the information from their newly arrived private observations. This can be achieved by considering the following regularized optimization problem:
\begin{align}\label{eq:fusion_rule_objective}
\min_{\psi\in\Delta_H}\Big\{D_{\text{KL}}(\psi||\bmeta_{k,i})-\gamma \,\e_{\psi}\log L_k(\bxi_{k,i}|\theta_i)\Big\}
\end{align}
where $\Delta_H$ is the probability simplex of dimension $H$, and $\e_{\psi}$ is the expectation computed with respect to $\psi$, i.e.,
\begin{align}
    \e_{\psi} \log L_k(\bxi_{k,i}|\theta_i) \triangleq \sum_{\theta_i \in \Theta} \psi (\theta_i) \log L_k(\bxi_{k,i}|\theta_i).
\end{align}
The objective function in \eqref{eq:fusion_rule_objective} consists of two terms. The first term is the KL-divergence term that penalizes the disagreement with the time-adjusted prior \( \{ \bmeta_{k,i} \} \). The second term corresponds to the log-likelihood of the observation \( \bxi_{k,i} \) averaged over the hypotheses with respect to \( \psi \). The cost in \eqref{eq:fusion_rule_objective} then seeks to minimize disagreement with the prior while maximizing the likelihood of the observation; the two terms are coupled by a regularization parameter $\gamma>0$. As we show in the sequel, different special cases of the problem setting might necessitate different \( \gamma \) values, e.g., \( \gamma =K, \gamma=1 \). Therefore, we continue with a general parameter \( \gamma >0 \). The objective function in \eqref{eq:fusion_rule_objective} can be expanded as
\begin{align}\label{eq:adaptation_obj_expansion}
   \sum_{\theta_i \in \Theta}  \psi(\theta_i) &\Big ( \log \frac{\psi(\theta_i)}{\bmeta_{k,i}(\theta_i)} -\gamma \log L_k(\bxi_{k,i}|\theta_i) \Big) \notag \\
   &=\sum_{\theta_i \in \Theta}  \psi(\theta_i) \log \frac{\psi(\theta_i)}{\bmeta_{k,i}(\theta_i)  (L_k(\bxi_{k,i}|\theta_i))^\gamma}
\end{align}
The RHS of \eqref{eq:adaptation_obj_expansion} is a KL-divergence under a proper normalization. Minimizing it results in the following \emph{local} \( \gamma\)-scaled Bayesian adaptation step for each agent:
\begin{align}\label{eq:dif_adapt_step_1}
 \bpsi_{k,i} (\theta_{i}) &\propto (L_k(\bxi_{k,i} | \theta_{i}))^{\gamma}\bmeta_{k,i} (\theta_{i}) \qquad \text{(Adapt)}
\end{align}
where \( \gamma > 0 \) scales the likelihood of the new observation against prior information. After agents independently obtain their intermediate beliefs $\bpsi_{k,i}$ according to \eqref{eq:dif_evolve_step_1} and \eqref{eq:dif_adapt_step_1}, they exchange these beliefs with their neighbors. Each agent $k$ will then need to fuse the beliefs received from the neighbors, and one way to do so is to seek the belief vector $\mu$ that solves~\cite{koliander2022,heskes97}:
\begin{align}\label{eq:fusion_objective_2}
\min_{\mu\in\Delta_H}\left\{\sum_{\ell\in\mathcal{N}_k}a_{\ell k}D_{\text{KL}}(\mu||\bpsi_{\ell,i})\right\}.
\end{align}
This objective function penalizes the average disagreement with the neighbors' intermediate beliefs and it can be expanded as
\begin{align}\label{eq:fusion_obj_expansion_2}	
\sum_{\theta_i\in\Theta}\mu(\theta_i)\sum_{\ell\in\mathcal{N}_k}&a_{\ell k}\log\frac{\mu(\theta_i)}{\bpsi_{\ell,i}(\theta_i)}\notag\\
&=\sum_{\theta_i\in\Theta}\mu(\theta_i)\log\frac{\mu(\theta_i)}{\prod_{\ell\in\mathcal{N}_k}[\bpsi_{\ell,i}(\theta_i)]^{a_{\ell k}}}
\end{align}
where the term on the RHS of \eqref{eq:fusion_obj_expansion_2} can be seen as a KL divergence under proper normalization, whose minimizer is given by the following geometric-average combination:
\begin{align}\label{eq:dif_combine_step_1}
  \bmu_{k,i}(\theta_{i}) &\propto \prod_{\ell \in \mathcal{N}_k}  \big ( \bpsi_{\ell,i} (\theta_{i}) \big )^{a_{\ell k}}   \quad \text{(Combine)}.
\end{align}
Exchanging and combining the beliefs repeatedly allows the local information to diffuse through the network. The complete procedure leads to the {\em diffusion HMM strategy} (DHS), which is listed in \eqref{eq:dif_evolve_step}--\eqref{eq:dif_combine_step}.

\begin{algorithm}[]
 \caption{{Diffusion HMM strategy (DHS)}}
 \begin{algorithmic}[1]
  \item set initial beliefs $\mu_{k,0}(\theta)>0$, $\forall k\in \mathcal{N}$ and $\forall\theta\in\Theta$ 
  \item choose $\gamma>0$
 \While{$i\geq 1$}
 \FORA{each agent $k\in\mathcal{N}$}
\FORB{each hypothesis $\theta_i\in\Theta$ }
 \begin{equation}\label{eq:dif_evolve_step}
  \bmeta_{k,i} (\theta_{i}) = \sum_{\theta_{i-1} \in \Theta} \bT(\theta_{i}| \theta_{i-1}) \bmu_{k,i-1}(\theta_{i-1}) 
\end{equation}
\ENDFORB
\State agent $k$ observes $\bxi_{k,i}$
\FORC{each hypothesis $\theta_i\in\Theta$ }
\begin{equation}\label{eq:dif_adapt_step}
    \bpsi_{k,i} (\theta_{i}) \propto (L_k(\bxi_{k,i} | \theta_{i}))^{\gamma}\bmeta_{k,i} (\theta_{i})  
\end{equation}
\ENDFORC
\ENDFORA
\FORD{each agent $k\in\mathcal{N}$ and each hypothesis $\theta_i\in\Theta$}
\begin{equation}\label{eq:dif_combine_step}
    \bmu_{k,i}(\theta_{i}) \propto \prod_{\ell \in \mathcal{N}_k}  \big ( \bpsi_{\ell,i} (\theta_{i}) \big )^{a_{\ell k}}    
\end{equation}
\ENDFORD
\State \( i \leftarrow i+1 \)
\EndWhile
 \end{algorithmic} \label{algorithm1}
 \end{algorithm}

The proposed DHS algorithm can be seen as a generalization of the following special cases:
\begin{itemize}
    \item When the network consists of a single-agent, i.e., \( K = 1\), the strategy is equivalent to the traditional optimal Bayesian filtering algorithm \cite[Chapter 3]{krishnamurthy_2016} when the local updates are Bayesian, i.e., when \( \gamma=1 \).
    \item If \( \gamma=1 \), and the true hypothesis is \emph{fixed}, i.e.,
\begin{align}\label{eq:fixed_truth_model}
\bT(\theta_i | \theta_{i-1}) = \begin{cases} 
      1, & \theta_i=\theta_{i-1} \\
      0, & \theta_i\neq\theta_{i-1}
   \end{cases},
\end{align} 
then, the algorithm reduces to the canonical log-linear social learning algorithms \cite{nedic_2017, lalitha_2018,matta_2020}. 
\item The beliefs of agents will match the optimal centralized belief \eqref{eq:centralized_posterior} \emph{exactly}, if the network is fully-connected with \( a_{\ell k} = 1/K \) \(\forall \ell,k\in\mathcal{N}\), all initial priors are equal (\(\mu_0^\star = \mu_{k,0}\), \(\forall k \in \mathcal{N}\)), and \( \gamma = K \). This conclusion follows by induction. First, assume that for each agent \( k \), \( \bmu_{k,i-1}=\bmu_{i-1}^\star \). This would imply that \(\bmeta_{k,i}=\bmeta_{i}^\star \) by the equivalence of the time-adjustment steps \eqref{eq:ck_cent} and \eqref{eq:dif_evolve_step}. Combining the adapt \eqref{eq:dif_adapt_step} and combine \eqref{eq:dif_combine_step} steps, the belief of agent \( k \) at time \( i \) then becomes
\begin{align}\label{eq:matching_centralized}
    \bmu_{k,i}(\theta_{i}) &\propto \prod_{\ell \in \mathcal{N}_k}  \big (L_k(\bxi_{k,i} | \theta_{i})\big )^{\gamma a_{\ell k}}\big (\bmeta_{k,i} (\theta_{i})  \big )^{a_{\ell k}}  \notag \\
    &\propto \bmeta_{i}^\star (\theta_{i})   \prod_{\ell=1}^K  L_k(\bxi_{k,i} | \theta_{i})
\end{align}
which is equivalent to the centralized update \eqref{eq:centralized_posterior}. Since the base case \( \mu_0^\star = \mu_{k,0}\) also holds, we conclude by induction that, the beliefs at all iterations will match the centralized belief.
\end{itemize}
We continue with the general strategy \eqref{eq:dif_evolve_step}--\eqref{eq:dif_combine_step}. Being motivated by different \( \gamma \) values for different special cases, we use a general step-size \( \gamma > 0\). We will assume the following condition on the initial beliefs to make sure  that no states are discarded from consideration.
\begin{assumption}[{\bf Initial priors} \cite{bandyopadhyay_2018,nedic_2017}]\label{as:positive_initial_beliefs}
All initial beliefs are strictly positive at all hypotheses, i.e., for each hypothesis \( \theta \in \Theta \) and for each agent \( k \), \( \mu_{k,0} (\theta) > 0 , \mu_0^\star(\theta) > 0\).
\qed
\end{assumption}
Also, to avoid pathological cases, we assume a finite KL-divergence condition between likelihoods. Namely, for any agent \( k \) and hypotheses \( \theta, \theta^\prime \in \Theta\) it holds that
\begin{align}
    \dkl ( L_k (\cdot| \theta) || L_k (\cdot | \theta^\prime) ) < \infty.
\end{align}
\noindent This condition ensures that the likelihoods at every agent for all hypotheses share the same support.

\subsection{Alternative Algorithms}\label{sec:alternative_algorithms}
The strategy proposed in Section \ref{sec:dif_HMM} is a \emph{diffusion}-based algorithm with \emph{geometric} averaging (GA), a.k.a. logarithmic opinion pooling. It assumes knowledge of the transition kernel \( \bT \).  One can also consider the following variations:
\begin{itemize}
    \item \textbf{Diffusion-AA:} In this case, step \eqref{eq:dif_combine_step} is replaced by 
    \begin{align}\label{eq:dif_aa_combine_step}
   \bmu_{k,i}(\theta_{i}) &= \sum_{\ell \in \mathcal{N}_k} a_{\ell k} \bpsi_{\ell,i} (\theta_{i}) 
    \end{align}
    where \emph{arithmetic} averaging (AA) is used in place of GA. In \cite{kayaalp2022_aa_ga}, it is shown that GA outperforms AA in terms of convergence rate in the fixed hypothesis case. \\
    \item \textbf{Consensus-GA:} In diffusion, agents exchange and combine the \emph{updated} intermediate beliefs \(\{\bpsi_{\ell,i}\}\). In contrast, in consensus \cite{jadbabaie_2012}, agents combine their intermediate belief with the neighbors' belief {\em prior} to updating, such as replacing \eqref{eq:dif_combine_step} by
        \begin{align}\label{eq:consensus_ga}
   \bmu_{k,i}(\theta_{i})  \propto  \big( \bpsi_{k,i} (\theta_{i}) \big)^{a_{k k}}  \prod\limits_{\ell \in \mathcal{N}_k \setminus \{k\}}  \big ( \bmeta_{\ell,i} (\theta_{i}) \big )^{a_{\ell k}}.
\end{align}
    \item \textbf{Adaptive social learning} (ASL)  \cite{bordignon_2021}: The combination step in ASL is also given by \eqref{eq:dif_combine_step}. However, agents do not utilize the transition model. Specifically, there is no evolution step \eqref{eq:dif_evolve_step} and the adaptation step \eqref{eq:dif_adapt_step} is modified to
    \begin{align}\label{eq:asl_adapt_step}
 \bpsi_{k,i} (\theta_{i}) \propto L_k(\bxi_{k,i} | \theta_{i})(\bmu_{k,i-1} (\theta_{i}))^{1-\delta},
\end{align}
where \( 0< \delta < 1 \) is a design parameter. The purpose of the parameter \( 1-\delta \) is to act as a forgetting factor to endow the ASL algorithm with the ability to track drifts in the state under non-stationary conditions. In contrast, in the proposed DHS strategy, tracking is done via the evolution step \eqref{eq:dif_evolve_step}. 
\end{itemize}

We compare these algorithms against the proposed strategy \eqref{eq:dif_evolve_step}--\eqref{eq:dif_combine_step} in Section \ref{sec:simulations}.

\section{Optimality Gap}\label{sec:optimality_gap}

In this section, we analyze the disagreement between the diffusion HMM strategy \eqref{eq:dif_evolve_step}--\eqref{eq:dif_combine_step} and the centralized solution \eqref{eq:centralized_posterior}--\eqref{eq:ck_cent}. For this section alone, we assume a regularity condition on the likelihood functions for technical reasons (similar to what was done in \cite{shahrampour_2016}).
\begin{assumption}[{\bf Regularity condition}]\label{as:likelihood_functions}
The absolute log-likelihood functions are uniformly bounded over their support for all agents:
\begin{align}
   \big |\log L_k (\xi | \theta) \big | \leq C_L, \qquad \forall k \in \mathcal{N}, \theta \in \Theta .
\end{align}
\qed
\end{assumption}
\noindent Assumption~\ref{as:likelihood_functions} implies that the likelihood functions do not get arbitrarily close to zero or arbitrarily large in their support. This ensures that each private signal \( \xi\) has bounded informativeness. For example, discrete signal space models or truncated Gaussian likelihoods satisfy this assumption.

Next, we introduce conditions on the transition model, which will play a crucial part in the analysis.

\subsection{Transition Model}

We assume that the transition Markov chain is \emph{irreducible} and \emph{aperiodic} \cite[Chapter 2]{resnick2002}. This means that there exists a constant integer \( n > 0\) such that for any two hypotheses \( \theta, \theta^\prime \in \Theta \):
\begin{align}
\bT^n (\theta | \theta^\prime) > 0,    
\end{align}
where \( \bT^n \) is the \(n\)-fold application of the transition kernel. This condition also implies that the Markov chain is \emph{ergodic} because the number of hypotheses \( H \) is finite \cite[Chapter 2]{resnick2002}. In other words, repeated application of the transition kernel \( \bT \) will converge to a limiting distribution regardless of the initial input distribution. More formally, for any input distribution \( \mu \in \Delta_H\),
    \begin{align}
    \lim_{n \to \infty} \sum_{\theta^\prime \in \Theta} \bT^n(\theta| \theta^\prime) \mu(\theta^\prime) = \pi (\theta), 
\end{align}
where \( \pi \) is the Perron vector of the \( H\times H\) transition matrix \( T \triangleq [\bT(\theta|\theta^\prime)]\). In this work, we consider the \emph{geometrically ergodic} \cite[Chapter~2]{krishnamurthy_2016} subclass of transition models. In the following, we define these models using the strong-data processing inequality (SDPI) \cite{polyanskiy_2017}. \\

\noindent \textbf{Strong-data processing inequality (SDPI)}\cite{polyanskiy_2017}: Consider any two discrete distributions over \( \Theta \), \( \mu^a \) and \( \mu^b \),  satisfying \( 0 < \dkl (\mu^a || \mu^b) < \infty \), and introduce their time-adjusted versions according to the Chapman-Kolmogorov equation as in \eqref{eq:dif_evolve_step}:
\begin{align}
    \eta^a (\theta_{i}) &= \sum_{\theta_{i-1} \in \Theta} \bT(\theta_{i}| \theta_{i-1}) \mu^a(\theta_{i-1}) 
\end{align}
and similarly for \( \eta^b \). Then, the SDPI states that:
\begin{align}
    \dkl (\eta^a || \eta^b ) \leq \kappa_{\text{KL}} (\bT) \dkl (\mu^a || \mu^b)
\end{align}
where \( \kappa_{\text{KL}} (\bT) \in [0,1] \) is a contraction coefficient defined as
\begin{align}
    \kappa_{\text{KL}} (\bT) \triangleq \sup_{\mu^a,\mu^b} \frac{\dkl (\eta^a || \eta^b ) }{\dkl (\mu^a || \mu^b )}.
\end{align}
Observe that the coefficient is only dependent on the transition model, and is not a function of the input distributions. An upper bound on \( \kappa_{\text{KL}} (\bT) \) is given by the Dobrushin's contraction coefficient \cite{dobrushin_1956,polyanskiy_2017,krishnamurthy_2016}, which is defined by 
\begin{align}
    \kappa (\bT) \triangleq \sup_{\theta^\prime,\theta^{\prime\prime} \in \Theta} \frac{1}{2} \sum_{\theta \in \Theta} \Big | \bT(\theta |\theta^\prime) - \bT(\theta |\theta^{\prime\prime}) \Big | \quad \in [0,1],
\end{align}
It is known that \cite{polyanskiy_2017}:
\begin{align}
     \kappa_{\text{KL}} (\bT) \leq \kappa (\bT), \qquad  \kappa_{\text{KL}} (\bT) = 1 \Longleftrightarrow \kappa (\bT) = 1 .
\end{align}
For example,
\begin{itemize}
    \item If the transition model is a binary symmetric channel:
\begin{equation}
\bT(\theta_i | \theta_{i-1}) = \begin{cases} 
      1-\alpha, & \theta_i=\theta_{i-1} \\
      \alpha, & \theta_i\neq\theta_{i-1}
   \end{cases},
\end{equation}
then \( \kappa(\bT) = | 1-2\alpha |\) \cite{polyanskiy_2017}. Notice that this is a symmetric function around the transition probability \( \alpha = 0.5 \), e.g., \( \alpha = 0.2 \) and \( \alpha = 0.8\) yield the same coefficient.
    \item The coefficient \( \kappa (\bT) = 0 \) if, and only if \cite[Ch. 2]{krishnamurthy_2016}, 
    \begin{equation}
        \bT(\theta_i | \theta_{i-1}) = \pi(\theta_i).
    \end{equation} 
    Note that this implies, for any \( \mu \in \Delta_H\),
    \begin{align}\label{eq:pi_transition}
    \sum_{\theta_{i-1} \in \Theta} \bT(\theta_{i}| \theta_{i-1}) \mu(\theta_{i-1}) &= \sum_{\theta_{i-1} \in \Theta} \pi(\theta_i) \mu(\theta_{i-1}) \notag \\ &=\pi(\theta_i)  \end{align}
    In other words, the transition kernel will output the same distribution \( \pi(\theta) \) regardless of the input distribution \( \mu(\theta)\). Observe that the rapidly mixing binary symmetric channel with transition probability \( \alpha = 0.5\) is an example of this case.
\end{itemize}
In general, for two input distributions \( \mu^a \) and \( \mu^b \), the output distributions resulting from an \(n\)-fold application of the transition kernel, i.e., 
\begin{align}
    \eta_n^a (\theta_{i}) &= \sum_{\theta_{i-n} \in \Theta} \bT^n(\theta_{i}| \theta_{i-n}) \mu^a(\theta_{i-n}),
\end{align}
and similarly for \( \eta_n^b \), satisfy the following SDPI:
\begin{align}\label{eq:sdpi_n_fold}
    \dkl (\eta_n^a || \eta_n^b ) \leq (\kappa (\bT))^n \dkl (\mu^a || \mu^b)
\end{align}
It is clear that if \( \kappa (\bT) < 1 \), then the disagreement between any two input distributions will approach 0 exponentially fast. Transition models for which \( \kappa (\bT) < 1 \) are said to be geometrically ergodic \cite[Ch. 2]{krishnamurthy_2016}. It is seen from \eqref{eq:sdpi_n_fold} that the coefficient \( \kappa (\bT)\) is a measure of how rapidly the initial conditions are forgotten. In particular, as \(\kappa (\bT) \to 0 \), forgetting is faster.
\begin{assumption}[{\bf Transition model}]\label{assumption:geometrically_ergodic}
The transition model \( \mathbb{T} \) is assumed to be geometrically ergodic, i.e., \( \kappa (\bT) < 1\). \qed
\end{assumption}

The class of geometric ergodic transition models comprises a large group of transition models. For instance, non-deterministic binary symmetric channels, i.e., with a transition probability \( \alpha \in (0,1)\), are geometrically ergodic. Moreover, transition matrices with all positive elements, and in general, those that satisfy a minorization condition \cite[Theorem 2.7.4]{krishnamurthy_2016} are examples of geometrically ergodic transition models. However, the geometrically ergodic class excludes some transition models such as the fixed hypothesis case, where \(\kappa (\bT) = 1\). We elaborate more on this issue in the sequel.

\subsection{Disagreement with the Centralized Strategy}

We introduce the following time-varying risks to compare the performance of the diffusion HMM strategy \eqref{eq:dif_evolve_step}--\eqref{eq:dif_combine_step} with the centralized solution \eqref{eq:centralized_posterior}--\eqref{eq:ck_cent}:
\begin{align}\label{eq:filter_risk}
J_{k,i}&\triangleq\e_{\f_i} \dkl(\bmu_i^\star || \bmu_{k,i}) 
\end{align}
and
\begin{align}\label{eq:prior_risk}
\widetilde{J}_{k,i}&\triangleq\e_{\f_{i-1}} \dkl(\bmeta_i^\star || \bmeta_{k,i}) 
\end{align}
where  \(\e_{\f_i}\) represents expectation over the distribution of \( \f_i \), which collects all observations from across the network until time $i$. Notice that the risks in \eqref{eq:filter_risk}--\eqref{eq:prior_risk} are not random variables, since the corresponding KL-divergences are averaged over all possible realizations of observations. The \emph{posterior} risk \( J_{k,i} \) in \eqref{eq:filter_risk} is the disagreement between the belief of agent \( k \) and the centralized belief at time \( i \), \emph{after} the observations in \( \bxi_i \) have been  emitted from that hypothesis. In comparison, the risk \( \widetilde{J}_{k,i} \) in \eqref{eq:prior_risk} is the divergence of time-adjusted \emph{priors}, which measures the disagreement \emph{before} the observations have been emitted. 


Our first result establishes that the disagreement between the centralized and distributed solutions is asymptotically bounded for all agents in the network.
\begin{theorem}[{\bf Asymptotic bounds}]\label{th:kl_without_net_assumption}For each agent \( k \),
under Assumptions \ref{as:independence}--\ref{assumption:geometrically_ergodic}, the risks \eqref{eq:filter_risk} and \eqref{eq:prior_risk} are asymptotically bounded, namely,
\begin{align}\label{eq:th1_risk}
\limsup_{i \rightarrow \infty} J_{k,i} \leq  \frac{2\sqrt{K} \gamma  \lambda C_L}{1-\kappa (\bT)}   
\end{align}
and
\begin{align}\label{eq:th1_prior}
\limsup_{i \rightarrow \infty} \widetilde{J}_{k,i} \leq  \frac{2\kappa(\bT)\sqrt{K} \gamma  \lambda C_L}{1-\kappa (\bT)}  
\end{align}
where \( \lambda \triangleq \max \{ |1-\frac{K}{\gamma}|, \rho_2 \} \), and \( \rho_2 \) is the second largest modulus eigenvalue of \( A \). 
\end{theorem}
\begin{proof}
See Appendix \ref{appendix:risk_theorem}.
\end{proof}

Swift and random changes in the environment can prevent a strongly-connected network from approaching the performance of the centralized solution close enough---especially when the network is sparse and it takes more time for the information to diffuse to all agents than the rate at which the state is changing. Therefore, the bounds in Theorem~\ref{th:kl_without_net_assumption} are not generally close to 0. The following remarks are now in place. Simulations that support these observations appear in Sec.\ref{sec:simulations}.
\begin{itemize}
\item \textbf{Matching the centralized strategy}: The bounds \eqref{eq:th1_risk} and \eqref{eq:th1_prior} will be tight when the distributed solution reduces to the centralized implementation, which happens with a uniformly weighted fully-connected network (\( a_{\ell k} = 1/K \) \(\forall \ell,k\in\mathcal{N}\), \( \rho_2 = 0 \)), same initial priors (\(\mu_0^\star = \mu_{k,0}\),\(\forall k \in \mathcal{N}\)), and \( \gamma = K \) as shown in \eqref{eq:matching_centralized}. In this case, the upper bounds will become zero as expected.
\item \textbf{Stability}: The bounds \eqref{eq:th1_risk} and \eqref{eq:th1_prior} are independent of the initial beliefs as long as Assumption \ref{as:positive_initial_beliefs} is satisfied. Indeed, geometric ergodicity is sufficient to asymptotically forget the initial conditions. In this way, the filter is robust to incorrect initializations.
\item \textbf{Network connectivity}: Note that as \( \gamma \to K \), the bounds become proportional to \( \rho_2 \), the mixing rate of the graph. This aspect emphasizes the benefit of cooperation. Highly-connected graphs, with small \( \rho_2 \), will be closer to the centralized solution while sparse networks or non-cooperative agents will have higher deviation.
\item \textbf{Network size}: The bounds are also proportional to the number of agents. The disagreement between the centralized solution and the individual agents increases with the square-root of the network size. Note that this does not mean agents would perform worse if new agents join the network. It is the \emph{relative} performance compared to the optimal solution, which has access to data from all agents, that could decrease.
\item \textbf{Ergodicity}: Notice that when \( \kappa (\bT) \to 0 \), from the bound  \eqref{eq:th1_prior}, it is obvious that \( \widetilde{J}_{k,i} \to 0 \). This is anticipated because if the coefficient \( \kappa (\bT) = 0\), the time-adjustment steps of the centralized and decentralized strategies \eqref{eq:ck_cent} and \eqref{eq:dif_evolve_step} become
\begin{align}\label{eq:two_time_adj_priors}
    \eta_i^\star (\theta_i) &=\!\!\!\sum_{\theta_{i-1}\in\Theta}\!\!\! \bT(\theta_{i}| \theta_{i-1}) \bmu_{i-1}^\star(\theta_{i-1}) \stackrel{\eqref{eq:pi_transition}}{=} \pi(\theta_i),  \notag \\
\eta_{k,i} (\theta_{i}) &=\!\!\! \sum_{\theta_{i-1} \in \Theta} \!\!\!\bT(\theta_{i}| \theta_{i-1}) \bmu_{k,i-1}(\theta_{i-1}) \stackrel{\eqref{eq:pi_transition}}{=} \pi(\theta_i), 
\end{align}
regardless of the input distributions \( \bmu_{i-1}^\star(\theta_{i-1}) \) and \( \bmu_{k,i-1}(\theta_{i-1})\). As a result,
\begin{align}
    \widetilde{J}_{k,i}&\stackrel{\eqref{eq:prior_risk}}{=}\e_{\f_{i-1}} \dkl(\bmeta_i^\star || \bmeta_{k,i}) 
    \stackrel{\eqref{eq:two_time_adj_priors}}{=}0.
\end{align}
Therefore, the bound \eqref{eq:th1_prior} captures the effect of the ergodicity of the transition model via the ergodicity coefficient \( \kappa (\bT)\) accurately. In particular, the results imply that the bound is tight for rapidly changing binary symmetric transition models where \( \kappa (\bT) \to 0 \).


\item \textbf{Informativeness of observations}: Nonetheless, the bounds fall short in capturing the effect of the observations. For example, if the true hypothesis is fixed, then the transition model \( \bT \) satisfies \eqref{eq:fixed_truth_model}. In this case, for two distributions \( \mu^a, \mu^b \in \Delta_H \),
    \begin{align}
    \eta^a (\theta_{i}) &= \sum_{\theta_{i-1} \in \Theta} \bT(\theta_{i}| \theta_{i-1}) \mu^a(\theta_{i-1}) \stackrel{\eqref{eq:fixed_truth_model}}{=}\mu^a(\theta_{i})
\end{align}
and similarly for \( \eta^b \). Then,
\begin{align}
    \dkl (\eta^a || \eta^b ) = \dkl (\mu^a || \mu^b),
\end{align}
and consequently, \( \kappa (\bT) = 1 \). Since this transition model is not geometrically ergodic, the bounds do not cover this case. However, it is known from the standard social learning literature \cite{jadbabaie_2012, zhao_2012,nedic_2017,lalitha_2018} that when the observations of the agents are informative enough, that is to say, the observations provide sufficient information about the underlying state\footnote{For several works, e.g., \cite{nedic_2017,lalitha_2018,bordignon_2021,kayaalp2022_aa_ga}, it is required that for each wrong hypothesis \( \theta \) there exists at least one clear-sighted agent \( k \) that has a positive KL-divergence between its likelihoods of the true hypothesis \( \theta^\circ\) and \( \theta \), i.e., \( \dkl (L_k (\xi | \theta^\circ) || L_k (\xi | \theta)) > 0 \).}, all agents will learn the true hypothesis eventually. In other words, beliefs of agents, as well as the belief of the centralized strategy, on wrong hypotheses go to 0. This means that log-beliefs on wrong hypotheses become degenerate, and hence, it is not clear how the KL-divergences in the risks \eqref{eq:filter_risk}--\eqref{eq:prior_risk} would behave in this situation, or whether the KL-divergence is a meaningful metric here. We leave examining this interesting regime of \( \kappa (\bT) \to 1 \), where the dominant factor is the informativeness of observations rather than the ergodicity, to future work.
\item \textbf{Single-agent case}: In fact, the distinction between the ergodicity of the transition model and the informativeness of the observations arises in the analysis of single-agent strategies as well. The stability of the algorithm would refer to the conditions under which a wrongly initialized belief converges to the true posterior distribution in HMMs. In the notation of this work, this is the special case of \( \gamma = K = 1 \) and \( \rho_2 = 0 \). Remember that in this special case the proposed diffusion HMM strategy is equivalent to a traditional optimal Bayes filter \cite[Ch. 3]{krishnamurthy_2016}. One criterion that the literature on the single-agent case uses is whether the total variation distance between the true posterior and the agent's belief vanishes in the mean asymptotically, i.e., whether (setting the single-agent index \( k=1\)),
\begin{equation}\label{eq:tv_stability_cond}
    \lim_{i \to \infty} \e_{\f_i} \Big [ \sum_{\theta \in \Theta} \big |\bmu_i^\star(\theta) - \bmu_{1,i}(\theta) \big | \Big ] \stackrel{?}{=} 0,
\end{equation}
when the initial belief of the agent is not accurate, i.e., when
\begin{equation}
    \mathbb{P} (\btheta_0^\circ = \theta ) = \mu_0^\star(\theta) \neq \mu_{1,0}(\theta).
\end{equation}
The total variation distance in \eqref{eq:tv_stability_cond} can vanish because of two mechanisms: either the observations are sufficiently informative about the true state, or the transition model is sufficiently ergodic. Even though there are some works that focus on the informativeness of observations \cite{chigansky2009intrinsic}, most of the results in the literature rely on ergodicity to establish stability \cite{shue98,mcdonald2020}. Similarly, Theorem~\ref{th:kl_without_net_assumption} also depends on the ergodicity of the transition model. In particular, reference \cite{mcdonald2020} studies the stability via the Dobrushin coefficient \( \kappa (\bT) \) and concludes that as long as \( \kappa (\bT) < 1/2 \), regardless of the observation model, \eqref{eq:tv_stability_cond} is satisfied. In comparison, our results are in terms of KL-divergences, but they can be expressed in terms of total variation distances if we use Pinsker's inequality \cite[Chapter 3]{csiszar11}. For the single-agent case, we get: 
\begin{align}
    \sum_{\theta \in \Theta} \big |\bmu_i^\star(\theta) - \bmu_{1,i}(\theta) \big | \leq& \big ( 2 \dkl (\bmu_i^\star || \bmu_{1,i}) \big )^{\frac{1}{2}} .
    \end{align}
   If we take the expectations, this relation implies
    \begin{align}
     \e_{\f_i} \Big [ \sum_{\theta \in \Theta} \big |\bmu_i^\star(\theta) - \bmu_{1,i}(\theta) \big | \Big ] &\leq \! \e_{\f_i} \big ( 2 \dkl (\bmu_i^\star || \bmu_{1,i})\! \big )^{\frac{1}{2}} \notag \\
      &\stackrel{(a)}{\leq}  \big ( 2 J_{1,i} \big )^{\frac{1}{2}} ,
      \end{align}
      where \( (a) \) follows from Jensen's inequality. Finally, taking the limit of both sides, we arrive at
      \begin{align}
\lim_{i \to \infty}\!\e_{\f_i} \Big [\!\sum_{\theta \in \Theta} \!\big |\bmu_i^\star(\theta) - \bmu_{1,i}(\theta) \big |\! \Big ]\!\!&\leq \sqrt{2} \lim_{i \to \infty}  \big ( J_{1,i} \big )^{\frac{1}{2}} \notag \\
&\stackrel{(b)}{=}\;0,
\end{align}
where \( (b) \) follows from the fact that the bound \eqref{eq:th1_risk} is equal to zero if \( \gamma = K = 1\) and \( \rho_2 = 0 \). This is a more general result, namely that the single-agent Bayesian filter is stable whenever \( \kappa (\bT) < 1\), as long as the informativeness of observations are bounded (Assumption~\ref{as:likelihood_functions}).



\end{itemize}

\section{Probability of Error and Convergence}\label{sec:probability_error}

The previous section analyzed the closeness of the diffusion HMM strategy to the optimal centralized solution. In this section, we study the error probability for each agent. There are a few results for the probability of error of HMM filtering, even in the single agent case. One notable result is \cite{shue_2001} where error recursions are obtained for the single-agent binary hypothesis setting case. In a similar spirit, we will obtain recursive equations for the probability density functions (pdfs) that capture the stochastic behavior of the underlying system. This will not only provide recursive formulas for the error probabilities, but will also allow us to deduce asymptotic properties for the diffusion HMM strategy. Similar to \cite{shue_2001}, for this section, we shift our focus to the binary hypothesis setting, i.e., throughout this section, we set \( H=2 \) and \( \Theta = \{ 0,1 \}\). We do not need to restrict ourselves to bounded log-likelihood signal models in this section, as was the case in Assumption \ref{as:likelihood_functions}. In this setting, the MAP-classifier at agent \( k \) at time instant \( i \) becomes
\begin{align}
    \widehat{\btheta}_{k,i} = \begin{cases}
    1,  \quad \text{if} \quad \bmu_{k,i} (1) > \bmu_{k,i} (0)\\
    0,  \quad \text{if} \quad \bmu_{k,i} (0) \geq \bmu_{k,i} (1)
    \end{cases}.
\end{align}
This estimator is equivalent to
\begin{align}
        \widehat{\btheta}_{k,i} = \begin{cases}
    1,  \quad \text{if} \quad \w_{k,i}  > 0\\
    0,  \quad \text{if} \quad \w_{k,i} \leq  0 
    \end{cases}
\end{align}
in terms of the log-belief ratio \( \w_{k,i} \) defined as
\begin{align}\label{eq:wki_definition}
    \w_{k,i} \triangleq \log \frac{\bmu_{k,i} (1)}{\bmu_{k,i} (0)}.
\end{align}
As such, the probability of error for agent \( k \) at time instant \( i \) is given by:
\begin{align}\label{eq:error_agent_def}
    p_{k,i} \triangleq \mathbb{P}(\btheta_i^\circ = 1,\w_{k,i} \leq 0)+\mathbb{P}(\btheta_i^\circ = 0, \w_{k,i} >0).
\end{align}
Let $f_{k,i}(\theta,{w}_k)$ denote the probability density function of the joint variable $\{\btheta_i^\circ, \w_{k,i}\}$ for agent $k$ at time $i$. Note that the joint variable $\{\btheta_i^\circ, \w_{k,i}\}$ mixes a discrete and a continuous random variable. Here, we use the general definition of density, i.e., density is the Radon-Nikodym derivative with respect to a measure. The corresponding measure for $\{\btheta_i^\circ, \w_{k,i}\}$ is the product measure of the counting measure for $\btheta_i^\circ$ and Lebesgue measure for $\w_{k,i}$ \cite{tao2011introduction}. For example, we can evaluate the probability that \( \btheta_i^\circ = \theta\) and $\w_{k,i}$ lies within the infinitesimal interval $({w}_k, {w}_k+d{w}_k)$ by computing
\begin{align}\label{eq:error_agent_dens}
    f_{k,i}(\theta,w_k)dw_k = \mathbb{P}(\btheta_{i}^\circ =\theta, \w_{k,i} \in (w_k,w_k+dw_k))
\end{align}
In this way, the probability of error \eqref{eq:error_agent_def} is given by
\begin{align}\label{eq:error_density}
     p_{k,i} &=\int_{w_k=- \infty}^{0} f_{k,i} (1,w_k)dw_k + \int_{w_k=0}^{\infty} f_{k,i} (0,w_k)dw_k.
\end{align}
We further consider the probability density function involving the log-beliefs across all agents, namely,
\begin{align}\label{eq:error_joint_dens}
     f_{i} (\theta,w)dw_1dw_2\cdots dw_K \triangleq \mathbb{P}(\btheta_{i}^\circ=\theta, \w_{i} \in (w,w+dw)),
\end{align}
in terms of the aggregate variables:
\begin{align}
     \w_{i} \triangleq \text{col}\{ \w_{\ell,i}\}_{\ell=1}^K,  \quad  w \triangleq \text{col}\{w_\ell\}_{\ell=1}^K .
\end{align}
If we integrate \eqref{eq:error_joint_dens} over all agents with the exception of agent \( k \) we can determine the marginal density for agent $k$, namely,
\begin{align}\label{eq:marginalization_agents}
     f_{k,i} (\theta,w_k) =\int\!\! ... \!\!\int  f_{i} (\theta,w)dw_1\cdots dw_{k-1}dw_{k+1}\cdots dw_K.
\end{align}
In what follows, we will derive a temporal recursion for the joint density given by \eqref{eq:error_joint_dens}, from which agent-specific densities can then be deduced. To this end, first observe that the diffusion equations \eqref{eq:dif_evolve_step}-\eqref{eq:dif_combine_step} can be written compactly in terms of the log-belief ratio:
\begin{align}\label{eq:logbelief_agent_rec}
    \w_{k,i} = \sum_{\ell \in \mathcal{N}_k} a_{\ell k} &\gamma \log \frac{L_{\ell}(\bxi_{\ell,i} | 1) }{L_{\ell}(\bxi_{\ell,i} | 0)} + \sum_{\ell \in \mathcal{N}_k} a_{\ell k} \log \frac{\bmeta_{\ell,i}(1)}{\bmeta_{\ell,i}(0)}\notag \\
    =\sum_{\ell \in \mathcal{N}_k} a_{\ell k} &\gamma \log \frac{L_{\ell}(\bxi_{\ell,i} | 1) }{L_{\ell}(\bxi_{\ell,i} | 0)}  \notag \\ &+ \sum_{\ell \in \mathcal{N}_k} a_{\ell k} \log \frac{\bT(1|0)+\bT(1|1) \text{exp}\{\w_{\ell,i-1}\}}{\bT(0|0)+\bT(0|1)\text{exp}\{\w_{\ell,i-1}\}}.
\end{align}
In matrix form, equation \eqref{eq:logbelief_agent_rec} leads to
\begin{align}\label{eq:diffusion_matrix_recursion}
 \w_{i} =  A^{\T} \bm{\nu}_{i} + A^{\T} \bm{\chi}_{i} \quad \text{(diffusion HMM)},
\end{align}
where we are introducing the vector of \(\gamma\)-scaled log-likelihood ratios (LLR) across the network:
\begin{align}
\bm{\nu}_{i} \triangleq \text{col} \Big \{ \gamma \frac{ L_{\ell}(\bxi_{\ell,i} | 1)}{L_{\ell}(\bxi_{\ell,i} | 0)} \Big \}_{\ell=1}^K,
\end{align}
and the vector of time-adjusted prior belief log-ratios across the network:
\begin{align}\label{eq:chi_def}
\bm{\chi}_{i} \triangleq \text{col} \Bigg \{  \log \frac{\bT(1|0)+\bT(1|1) \text{exp}\{\w_{\ell,i-1}\}}{\bT(0|0)+\bT(0|1)\text{exp}\{\w_{\ell,i-1}\} } \Bigg \}_{\ell=1}^K.
\end{align}
If the underlying distributed strategy is instead the consensus algorithm from \eqref{eq:consensus_ga} in lieu of the diffusion strategy \eqref{eq:dif_evolve_step}--\eqref{eq:dif_combine_step}, then \eqref{eq:diffusion_matrix_recursion} would be replaced by
\begin{align}\label{eq:consensus_matrix_recursion}
    \w_{i} = \bm{\nu}_{i} + A^{\T} \bm{\chi}_{i} \quad \text{(consensus HMM)}.
\end{align}
 Observe that the joint density over two consecutive time instants of  \( \{\btheta_{i-1}^\circ ,\w_{i-1},\btheta_{i}^\circ ,\w_{i} \} \) satisfies
\begin{align}\label{eq:joint_time_density}
   \mathbb{P}(\btheta_{i}^\circ=\theta, \w_{i} &\in (w,w+dw), \btheta_{i-1}^\circ =\theta^\prime,\notag \\& \qquad \qquad \qquad  \w_{i-1} \in (w^{\prime},w^{\prime}+dw^{\prime})) \notag\\ &=S_{i}^{(\theta)} (w,w^{\prime}) \bT(\theta| \theta^\prime) f_{i-1}(\theta^\prime,w^{\prime})  dW dW^{\prime} 
\end{align}
where we are using \( dW \triangleq dw_1dw_2\cdots dw_K \) and \( dW^{\prime} \triangleq  dw_1^{\prime}dw_2^{\prime}\cdots dw_K^{\prime}\) for notational brevity, and where we are introducing the conditional probability 
\begin{align}\label{eq:si_definition}
    &S_{i}^{(\theta)} (w,w^{\prime})dW \notag \\ &\qquad \triangleq \mathbb{P}(\w_{i} \in (w,w+dw)|\btheta_{i}^\circ =\theta , \btheta_{i-1}^\circ =\theta^\prime, \w_{i-1} = w^{\prime}) \notag \\
  &\qquad \stackrel{(a)}{=}\mathbb{P}(\w_{i} \in (w,w+dw)|\btheta_{i}^\circ =\theta , \w_{i-1} = w^{\prime})
\end{align}
where \( (a) \) follows from the fact that \( \w_{i} \) is a function of \( \bxi_i\) only, once \( \w_{i-1} \) and \( \btheta_{i}^\circ \) are given. Therefore, the log-belief ratio \( \w_{i} \) is conditionally independent of \( \btheta_{i-1}^\circ \) --- see \eqref{eq:diffusion_matrix_recursion} for diffusion and \eqref{eq:consensus_matrix_recursion} for consensus. Note that, for diffusion algorithms, in general, even under the independence Assumption \ref{as:independence}:
\begin{align}
      &S_{i}^{(\theta)} (w,w^{\prime})dW \notag \\
      &\neq \prod_{\ell=1}^K \mathbb{P}(\w_{\ell,i} \in (w_\ell,w_\ell+dw_\ell)|\btheta_{i}^\circ =\theta , \w_{i-1} = w^{\prime}) 
\end{align}
because the newly arrived data \( \bxi_{k,i}\) is utilized by agent \( k \) as well as by its neighbors in the same iteration. On the other hand, for consensus, under Assumption \ref{as:independence},
\begin{align}\label{eq:si_independence_consensus}
      &S_{i}^{(\theta)} (w,w^{\prime})dW \notag \\
      &= \prod_{\ell=1}^K \mathbb{P}(\w_{\ell,i} \in (w_\ell,w_\ell+dw_\ell)|\btheta_{i}^\circ =\theta , \w_{i-1} = w^{\prime}) 
\end{align}
since the fresh data is used by the observing agent only. In other words, in the consensus implementation \eqref{eq:consensus_ga}, each log-belief $\w_{k,i}$ is a function of that agent's observation $\bm{\xi}_{k,i}$ only (given $\btheta_{i}^\circ$ and $\w_{i-1}$) and does not depend on the observations at the other agents at that time. This distinction between diffusion and consensus can be seen from \eqref{eq:diffusion_matrix_recursion} and \eqref{eq:consensus_matrix_recursion}. 

Now, marginalizing \eqref{eq:joint_time_density} with respect to the state $\theta^\prime$ and the log-belief ratio $w^\prime$ at the earlier time instant yields the following temporal recursion for the joint density in \eqref{eq:error_joint_dens}:

\begin{align}
 &f_{i} (\theta,w) dW = \sum_{\theta^\prime} \bT(\theta |\theta^\prime) \Bigg [ \int_{w^{\prime}}    \mathbb{P}(\btheta_{i}^\circ=\theta, \w_{i} \in (w,w+dw), \notag \\& \qquad \qquad \quad \qquad  \btheta_{i-1}^\circ =\theta^\prime, \w_{i-1} \in (w^{\prime},w^{\prime}+dw^{\prime})) \Bigg ],
 \end{align}
 which implies that
 \begin{align}\label{eq:fialpha_recursion}
 &f_{i} (\theta,w) \! =\!\sum_{\theta^\prime} \bT(\theta| \theta^\prime) \Bigg [ \!\!\int_{w^{\prime}}   S_{i}^{(\theta)} (w,w^{\prime}) f_{i-1} (\theta^\prime,w^{\prime})dW^\prime \Bigg ].
\end{align}

To summarize, the probability of error at each time instant \( i \) can be computed by (i) using \eqref{eq:fialpha_recursion} to find the joint density at time \( i \), (ii) marginalizing the joint density to find the agent-specific density by \eqref{eq:marginalization_agents}, and finally (iii) integrating the agent density as in \eqref{eq:error_density}.

\begin{remark}[{\bf Evaluation of integrals}]
Finding closed-form expressions to the integral expressions (e.g., \eqref{eq:fialpha_recursion}, \eqref{eq:marginalization_agents}) might not be feasible. One can apply numerical integration methods such as Monte Carlo techniques \cite{kalos2009monte} to compute the desired integrals. \qed
\end{remark}
The analysis until here holds for general transition and likelihood models. However, the kernel \( S_{i}^{(\theta)} (w,w^{\prime}) \) might not be obtained in closed form in general. Therefore, for stronger results, we focus on Gaussian likelihood models in the next section.
\subsection{Gaussian Likelihoods}\label{sec:gaussian_lklhd}
Thus, let us consider now Gaussian models of the form:
\begin{equation}
\begin{aligned}\label{eq:gaussian_obs_models}
    L_k(\xi_{k,i}|\btheta_i^\circ = 1)&=\frac{1}{\sqrt{2\pi \sigma_k^2}}\text{exp} \Big \{ - \frac{(\xi_{k,i}-\zeta(1))^2}{2 \sigma_k^2} \Big \} \\
    L_k(\xi_{k,i}|\btheta_i^\circ = 0)&=\frac{1}{\sqrt{2\pi \sigma_k^2}}\text{exp} \Big \{ - \frac{(\xi_{k,i}-\zeta(0))^2}{2 \sigma_k^2} \Big \} 
\end{aligned}
\end{equation}
where the means are assumed to satisfy \( \zeta(0) = -\zeta(1) = \zeta \) for a constant value \( \zeta \neq 0  \), and where the agent-specific variances satisfy \( \sigma_k^2 > 0 \). The corresponding log-likelihood ratio appearing in \eqref{eq:logbelief_agent_rec} is then given by
\begin{align}
    \log \frac{L_k(\xi_{k,i}|\btheta_i^\circ = 1)}{L_k(\xi_{k,i}|\btheta_i^\circ = 0)} = \frac{-2 \zeta \xi_{k,i} }{\sigma_k^2},
\end{align}
and accordingly, the vector $\bm{\nu}_i$ of \( \gamma\)-scaled LLRs across agents is given by
\begin{align}
\bm{\nu}_{i}  = \text{col} \Big \{ \frac{-2 \gamma \zeta }{\sigma_{\ell}^2} \bxi_{\ell,i} \Big \}_{\ell=1}^K.
\end{align}
By Assumption \ref{as:independence}, the \( \{ \bxi_{\ell,i} \}_{\ell=1}^K\) are independent random variables conditioned on \( \btheta_{i}^\circ \). This implies that \( \bm{\nu}_{i} \) is a multivariate Gaussian random variable conditioned on the true hypothesis \( \btheta_{i}^\circ \) at time instant \( i \), 
\begin{align}\label{eq:v_gaussian}
    \bm{\nu}_{i} \big |_{\btheta_{i}^\circ} \sim \mathcal{G} \Big ( \bm{\beta}^{(\theta_{i}^\circ)} , \Sigma \Big )
\end{align}
with mean
\begin{align}\label{eq:v_gaussian_mean}
    \bm{\beta}^{(\theta_{i}^\circ)} \triangleq \text{col} \Big \{ \frac{-2 (-1)^{\btheta_{i}^\circ}  \gamma \zeta^2 }{\sigma_\ell^2}\Big \}_{\ell=1}^K,
\end{align}
and covariance matrix
\begin{align}\label{eq:covar_definition}
    \Sigma \triangleq \text{diag} \Big \{\frac{4\gamma^2 \zeta^2}{\sigma_\ell^2} \Big \}_{\ell=1}^K.
\end{align}
Next, we treat the consensus and diffusion cases separately. The consensus case is straightforward and useful to understand the diffusion case.
\subsubsection{Consensus}
Using \eqref{eq:consensus_matrix_recursion} and the distribution \eqref{eq:v_gaussian} for \( \bm{\nu}_{i} \), we conclude that the conditional pdf of $\w_i$ given the current state and the prior log-belief ratio vector $\w_{i-1}$ is also Gaussian and equal to
\begin{align}\label{eq:si_consensus_expression}
    &S_{i}^{(\theta)} (w,w^{\prime}) \notag \\ &= \frac{\text{exp} \Big \{ -\frac{1}{2} (w - \rho^{(\theta)}(w^{\prime}) )^{\T}  \Sigma^{-1} (w- \rho^{(\theta)}(w^{\prime}) )\Big \}}{\sqrt{(2\pi)^{K} \text{det}( \Sigma  ) }}  
\end{align}
where the mean is defined by
\begin{align}
   \rho^{(\theta)} (w^{\prime}) \triangleq  \bm{\beta}^{(\theta_{i}^\circ)} + A^{\T} \bm{\chi}_{i} \Big |_{\btheta_{i}^\circ = \theta, \w_{i-1} = w^{\prime}}.
\end{align}
Observe that since \( \Sigma \) is a diagonal matrix, Eq. \eqref{eq:si_consensus_expression} can also be written as the multiplication of individual conditional densities, as already suggested by \eqref{eq:si_independence_consensus}.

\subsubsection{Diffusion}

The covariance matrix \( \Sigma \) in \eqref{eq:covar_definition} is non-singular since it is diagonal with positive diagonal entries. Consequently, \( \bm{\nu}_i \) in \eqref{eq:v_gaussian} is a non-degenerate random variable. In the consensus implementation \eqref{eq:consensus_matrix_recursion}, the variable \( \bm{w}_i \)  is an additive shift of \(\bm{\nu}_i \) conditioned on \( \bm{w}_{i-1} \). Therefore, \( \bm{w}_i \) is also a non-degenerate random variable and it admits the conditional density \eqref{eq:si_consensus_expression}.

In diffusion, however, \( \bm{w}_i \) is an affine transformation of \(\bm{\nu}_i \)---see \eqref{eq:diffusion_matrix_recursion}. The combination matrix \( A \) need not be invertible and hence, \( \bm{w}_i \) might not admit a density in \( \mathbb{R}^K \) in general. In Appendix \ref{appendix:diffusion_recursion}, we show that by representing \( \bm{w}_i \) in an \(r\)-dimensional subspace, where \(r\) is the rank of \(A\), no information is lost and the analysis and conclusions can be adjusted accordingly.

\begin{remark}[{\bf Difference from \cite{shue_2001}}]
 In \cite{shue_2001}, the probability of error recursions are studied for the \emph{single-agent} case only. Moreover, the recursions are based on belief differences instead of log-belief ratios. In that case, transition kernels are not Gaussian even under Gaussian observation models, as opposed to \eqref{eq:si_consensus_expression} and \eqref{eq:si_q_definition}. \qed
\end{remark}

\subsection{Asymptotic Convergence}
In addition to providing a way for calculating the error probabilities, the density evolution recursion \eqref{eq:fialpha_recursion} also allows us to show that agents exhibit a regular behavior in the limit. In particular, the distributions of the beliefs \(\bmu_{k,i}\) and log-belief ratios \(\w_{k,i}\) will converge to the distribution of some time-independent random variables. That is to say, they will converge \emph{in distribution} \cite[Chapter 17]{williams_1991}. A sequence (over time index \( i \)) of random variables \( \bm{x}_i \) converges to a limiting random variable \( \bm{x} \) in distribution if it holds that
\begin{align}\label{eq:def_in_dist_conv}
    \lim_{i \to \infty}  \mathbb{P}(\bm{x}_{i}  \in \mathcal{X} ) = \mathbb{P}(\bm{x} \in \mathcal{X} )
\end{align}
for a set \( \mathcal{X} \) of \(  \bm{x} \), whose boundary has zero probability under the limiting distribution. We denote this by writing
\begin{align}
    \bm{x}_{i} \stackrel{d}{\rightsquigarrow} \bm{x}.
\end{align}
Although beliefs can demonstrate random behavior with fluctuations in the limit, convergence in distribution implies the existence of limiting statistics such as steady-state probability of errors, as shown in the following result. 

\begin{theorem}[{\bf Asymptotic probability of error}]\label{theorem:asymptotic_probability}
The diffusion and consensus HMM strategies \eqref{eq:dif_evolve_step}--\eqref{eq:dif_combine_step} and \eqref{eq:consensus_ga} are asymptotically stable (in the sense of \cite{lasota1998chaos}) under binary classification, Gaussian likelihood models \eqref{eq:gaussian_obs_models}, and non-deterministic transition models (i.e., \( \bT(\theta |\theta^\prime) < 1\) \( \forall \theta,\theta^\prime \in \Theta \)). That is, the density function \( f_{i} \) satisfy
\begin{align}
    & \lim_{i \to \infty} \| f_{i}  - f_{\infty}  \|_{\textup{TV}} = 0,
\end{align}
 where \( f_{\infty}\) is a unique stationary density, and \( \| \cdot \|_{\textup{TV}} \) is the total variation norm defined with respect to the product measure, i.e., for any two densities \( f \) and \( g\)
 \begin{equation}
     \| f  - g  \|_{\textup{TV}} \triangleq \frac{1}{2} \sum_{\theta \in \Theta} \int  \big | f (\theta, w) - g (\theta, w)\big | dW.
 \end{equation}
 This result implies the convergence of the distribution of the log-belief ratios, and as a special case, the agent-specific probability of errors converge as well:
\begin{align}
    \w_{k,i} \stackrel{d}{\rightsquigarrow} \w_{k,\infty}, \qquad \lim_{i \to \infty} p_{k,i} = p_{k,\infty}.
\end{align}
\end{theorem}
\begin{proof}
See Appendix \ref{appendix:asymptotic_probability}.
\end{proof}
In order to establish this result, we employ in Appendix \ref{appendix:asymptotic_probability} a known result from ergodic theory \cite[Theorem 5.7.4]{lasota1998chaos}. As is common in ergodic theory, even though we can affirm that there exist limiting distributions; we do not know exactly what these distributions are. This task might be of formidable complexity in general.

Note that this result is in contrast to fixed-hypothesis social learning, where log-belief ratios do not converge, i.e., \(  \w_{k,i} = \log \frac{\bmu_{k,i} (1)}{\bmu_{k,i} (0)} \to -\infty \)  (assuming w.l.o.g. that \( \theta = 0 \) is the fixed true hypothesis).

Theorem \ref{theorem:asymptotic_probability} states that the agents' probability of error will approach a steady-state value despite the fact that the belief vectors can fluctuate randomly in the limit.

\begin{corollary}[{\bf Asymptotic beliefs}]\label{cor:beliefs_risks}
Theorem~\ref{theorem:asymptotic_probability}  implies that beliefs of the agents converge in distribution. More formally,
\begin{align}
    \bmu_{k,i} \stackrel{d}{\rightsquigarrow} \bmu_{k,\infty},
\end{align}
for a time-independent random variable \( \bmu_{k,\infty} \).
\end{corollary}
\begin{proof}
By definition \eqref{eq:wki_definition}, 
\begin{align}
    \w_{k,i} = \log \frac{1-\bmu_{k,i}(0)}{\bmu_{k,i}(0)} \Longleftrightarrow \bmu_{k,i}(0) = \frac{1}{1+ \text{exp}\{ \w_{k,i}\}}.
\end{align}
Since this is a continuous and non-degenerate transformation, by Theorem~\ref{theorem:asymptotic_probability} and the continuous mapping theorem \cite{mann1943} it holds that \( \bmu_{k,i} \stackrel{d}{\rightsquigarrow} \bmu_{k,\infty},\) for some time-independent random variable \(\bmu_{k,\infty}\). This implies that the statistics of the belief distribution also converge. 
\end{proof}
Corollary~\ref{cor:beliefs_risks} suggests that in general, the beliefs of agents will have random characteristics and will fluctuate in the long-run. This is in contrast to conventional social learning models where beliefs on the true \emph{fixed} hypothesis converge to one almost surely \cite{jadbabaie_2012,nedic_2017,zhao_2012,lalitha_2018}. In other words, all agents come to an agreement on the truth eventually. In comparison, in the current \emph{dynamic} hypothesis scenario, agents do not even come to an agreement as shown in the next result.

\begin{lemma}[{\bf Network disagreement}]\label{prop:network_disagreement} 
In general, the agents' beliefs do not converge to the same random variable in distribution. Namely, for any agent pair \( (\ell,k )\), the limiting variables \( \bmu_{k,\infty} \) and \( \bmu_{\ell,\infty} \) are not necessarily distributed according to the same distribution. Moreover, agents will have different performance in the long run, namely:
\begin{align}
    p_{k,\infty} \neq p_{\ell, \infty}, \qquad J_{k,\infty} \neq J_{\ell,\infty}.
\end{align}
\end{lemma}
\begin{proof}
We prove this by a counter-example in Appendix~\ref{sec:disagreement_example}.
\end{proof}
Lemma~\ref{prop:network_disagreement} implies that rapidly changing states prevent learning the truth with full confidence, as well as eventual network agreement, even under strongly-connected networks where information can flow thoroughly in all directions. Moreover, when the true state of nature is changing, agents can have different and non-vanishing asymptotic error probabilities. In traditional social learning, agents can have different and non-zero error probabilities in \emph{finite-time}. But as time goes by, all probabilities of error vanish, i.e., they all become 0. So, unlike the traditional setting, the dynamic truth model gives rise to an equilibrium of \emph{wise} and \emph{unwise} agents in asymptotics. That is to say, some agents will be more successful in predicting the truth than others in steady-state. The agents' error probabilities will be dependent on their observations' informativeness and their relative location in the network. Indeed, this ``wise agent phenomenon'' is more in line with what we observe in real-world, as against to eventual agreement of agents on the correct hypothesis that traditional social learning literature concludes. This observation shows the importance of incorporating the changing behavior of the state of nature into social learning models.

As discussed before, most of the literature on learning over strongly-connected social networks conclude consensus across agents, although there are exceptions. The works \cite{acemoglu13} and \cite{yildiz13} show that when there are \emph{stubborn} agents in the network that never change their opinion, the beliefs in the long-run can fluctuate and vary, as in the current work. Moreover, if agents tend to communicate with other agents that think alike \cite{axelrod97,blondel09}; or if they tend to use their own beliefs as substitutes for others' beliefs in the case of limited communication \cite{kayaalp2022random}, then opinion clusters can emerge. The current work gives evidence for another reason of disagreement, namely, the rapidly changing truth.  



\section{Simulation Results}\label{sec:simulations}
We consider the $10-$agent network displayed in Fig.~\ref{fig:net}. The combination weights are given by the Metropolis rule \cite{  metropolis_1953,sayed_2014}, which results in a doubly-stochastic and symmetric matrix with the mixing rate $\rho_2 = 0.86$, satisfying Assumption \ref{as:network_top}. 
\begin{figure}[ht]
	\centering
	\includegraphics[width=2.5in]{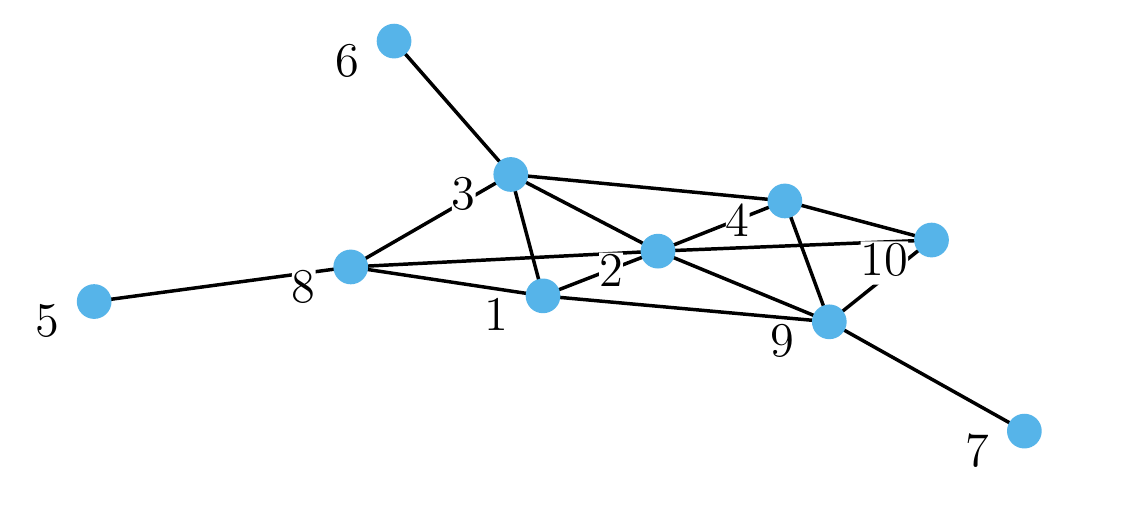}
	\caption{The network topology.}
	\label{fig:net}
\end{figure}

The agents over the network aim to track the true state of nature from a set of two hypotheses,  $\Theta=\{0,1\}$. For the initial simulations, all agents possess the same family of truncated Gaussian likelihoods, satisfying Assumption \ref{as:likelihood_functions}:
\[ L_k(\xi|\theta)  = \begin{dcases} 
       \frac{1}{Z_{\theta}}\frac{1}{\sqrt{2\pi}}\exp\left\{-\frac{(\xi - (1.5\times\theta)^2}{2}\right\},\!\!\!\!\!&-1\leq \xi \leq 2 \\
      \qquad \qquad 0, &  \text{otherwise}
   \end{dcases}
\]
for each agent $k\in \mathcal{N}$, where \( Z_{\theta} \) is a normalization constant:
\begin{align}
    Z_{\theta} \triangleq \int_{-1}^2 \frac{1}{\sqrt{2\pi}}\exp\left\{-\frac{(\xi - (1.5\times\theta))^2}{2}\right\} d\xi
\end{align}
The observations are independent conditioned on the true state, satisying Assumption \ref{as:independence}. The hidden state is changing with respect to the following transition model:
\begin{align} \bT(\theta_i | \theta_{i-1}) = \begin{cases} 
      0.9, & \theta_i=\theta_{i-1} \\
      0.1, & \theta_i\neq\theta_{i-1}
   \end{cases}
\end{align}
for which the Dobrushin coefficient is given by \( \kappa(\bT) = 0.8 \).

\begin{figure}[ht]
	\centering
	\includegraphics[width=2.7in]{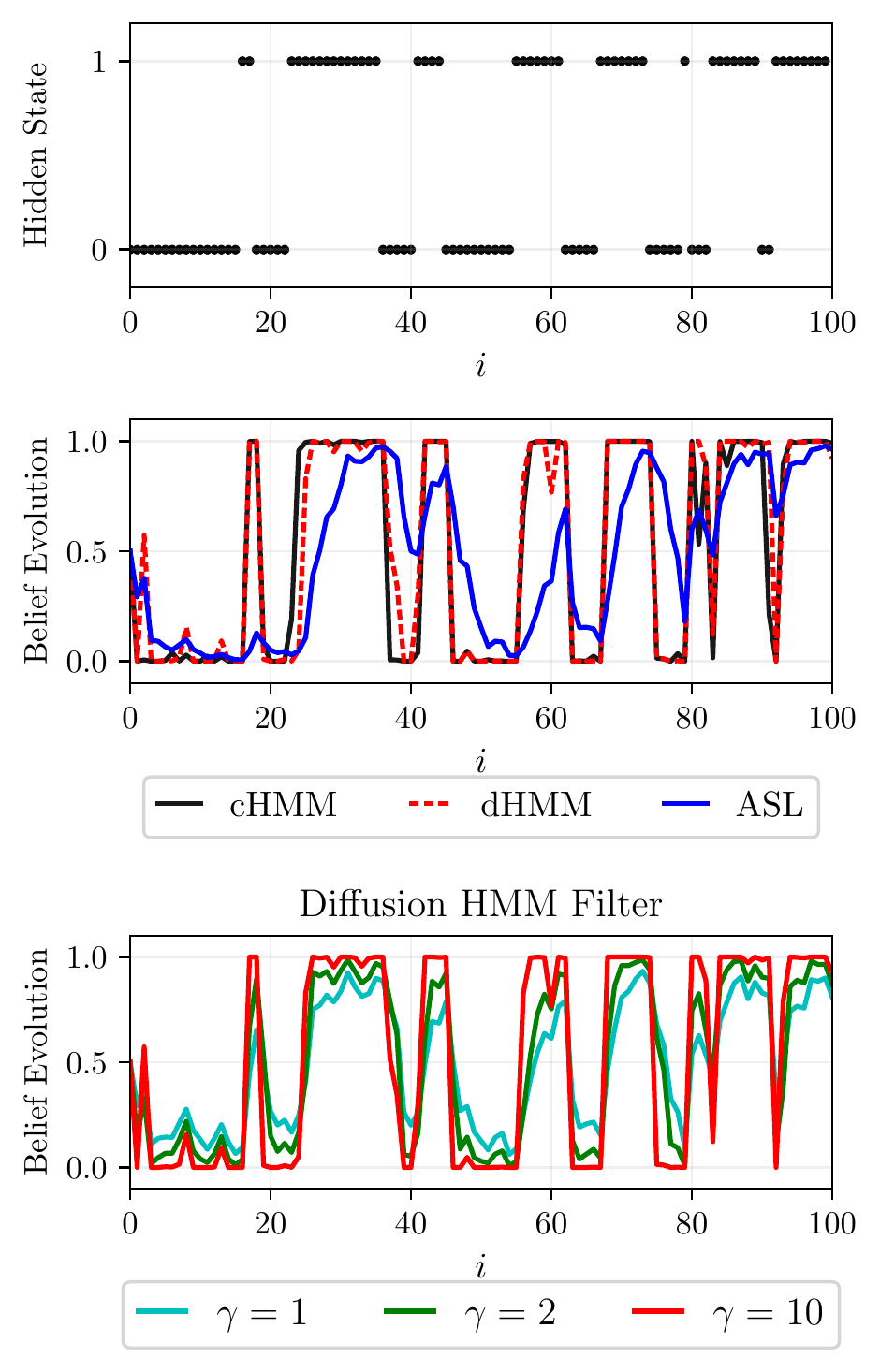}
	\caption{{\em Top panel}: A realization of the true hidden state. {\em Middle panel}: Belief evolution over time for different algorithms (cHMM, dHMM, and ASL \cite{bordignon_2021}). {\em Bottom panel}: Belief evolution over time for different \( \gamma \).}
	\label{fig:one_realization}
\end{figure}

The top panel of Fig.~\ref{fig:one_realization} demonstrates a particular realization of hidden states $\bm{\theta}^{\circ}_{i}$. In the middle panel, the belief evolution under this realization is shown for the following algorithms: the proposed diffusion HMM filter (dHMM) with the choice $\gamma=K$, the centralized HMM filter (cHMM), and ASL \cite{bordignon_2021} with the choice of \( \delta=\alpha=0.1\). Notice that dHMM and cHMM behave similarly with a remarkable performance for tracking the abrupt changes in the true state. They are faster in responding to state changes compared to ASL, which does not utilize knowledge of the transition model.


The bottom panel of Fig.~\ref{fig:one_realization} provides the belief evolution over time for different choices of \( \gamma \) in the diffusion HMM filter. As \(\gamma\) gets closer to \(K=10\), we can see that the tracking capacity of the algorithm increases, approaching the centralized algorithm.

The evolution of different agents' risk functions \( J_{k,i} \) over time is provided in Fig.~\ref{fig:risks_dif_agents}. Although they all exhibit a regular and bounded behavior as suggested by Theorem \ref{th:kl_without_net_assumption}, they are different with respect to different agents. More central agents have less divergence from the optimal centralized solution as expected, whereas marginal agents, such as agents 5, 6 and 7 present higher divergence.

\begin{figure}[ht]
	\centering
	\includegraphics[width=0.8\linewidth]{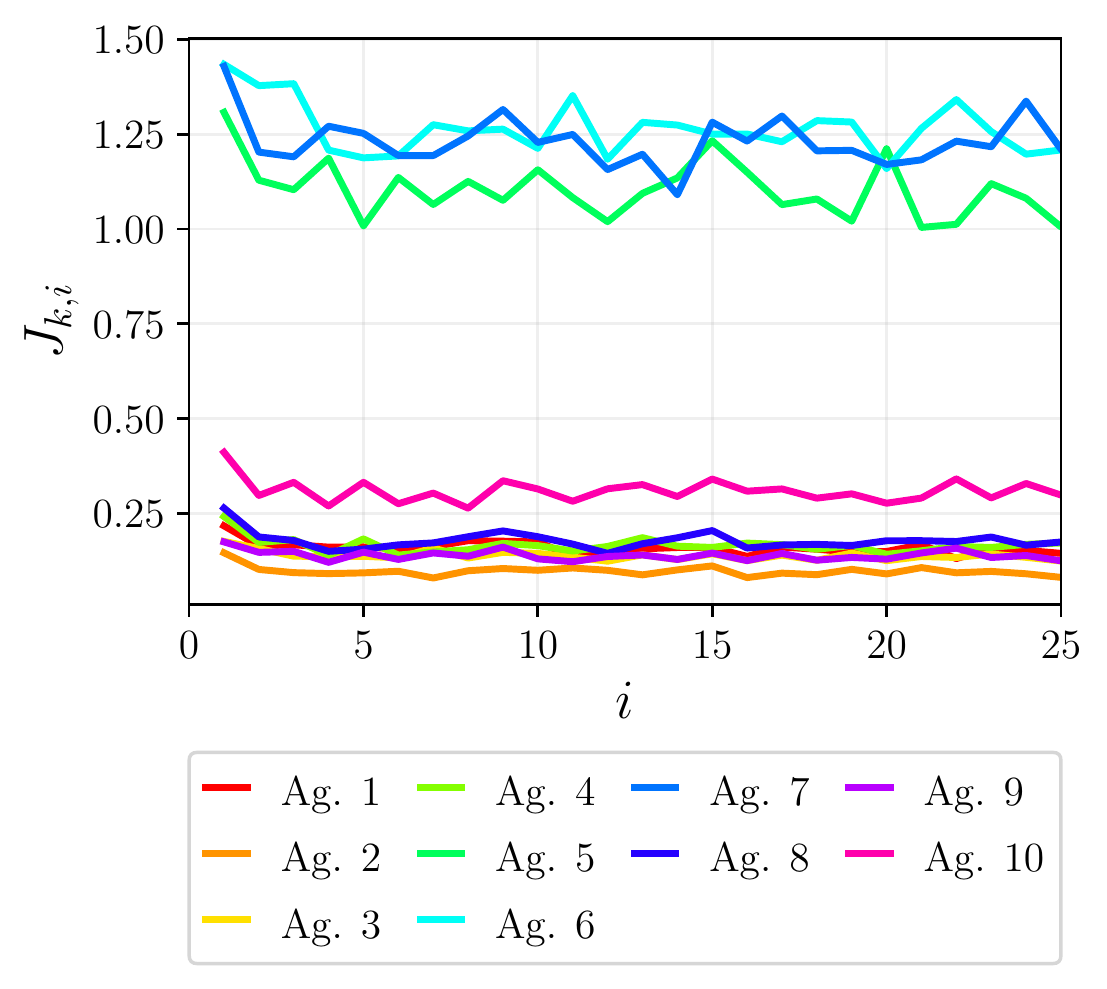}
	\caption{Risk functions over time belonging to different agents.}
	\label{fig:risks_dif_agents}
\end{figure}

\begin{figure*} 
    \centering
  \subfloat[\label{fig:risk_rho2}]{%
       \includegraphics[width=0.3\linewidth]{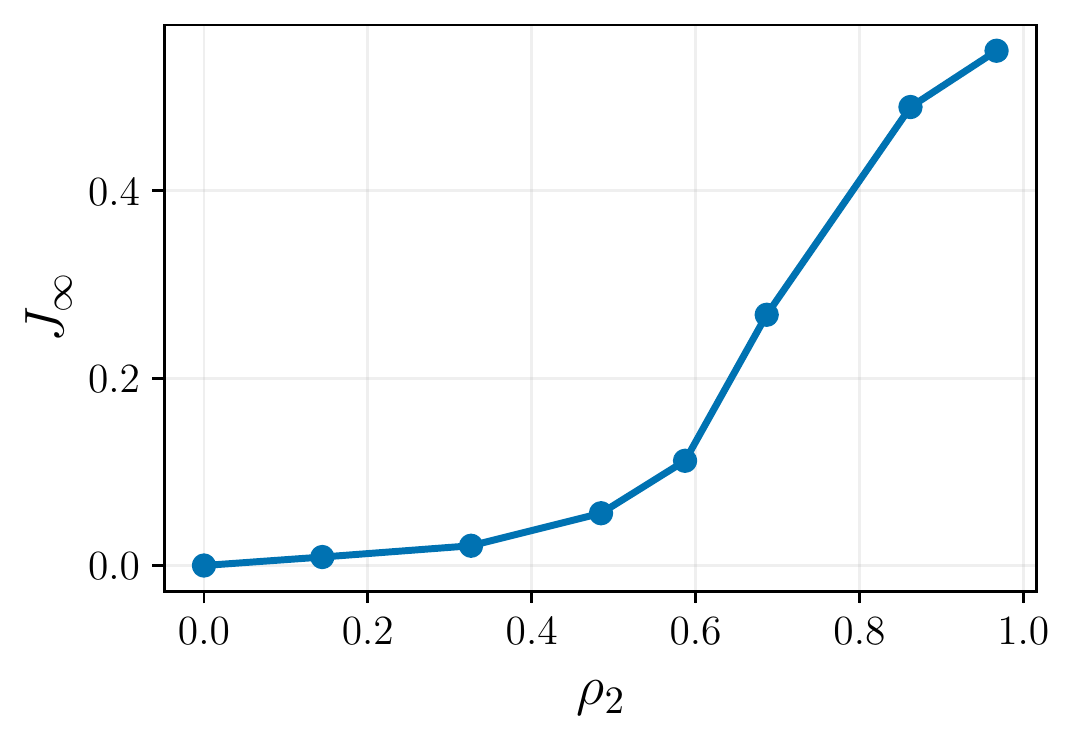}}
    \hfill
  \subfloat[\label{fig:risk_kappa}]{%
        \includegraphics[width=0.3\linewidth]{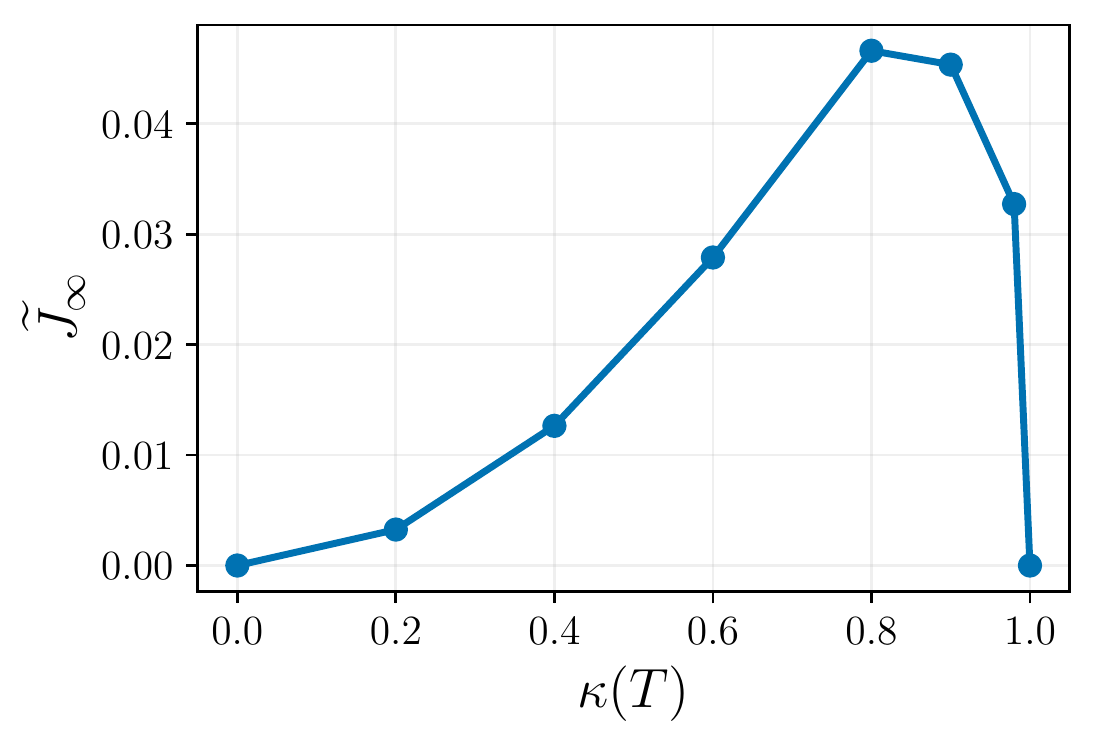}}
        \hfill
  \subfloat[\label{fig:risk_comb}]{%
       \includegraphics[width=0.3\linewidth]{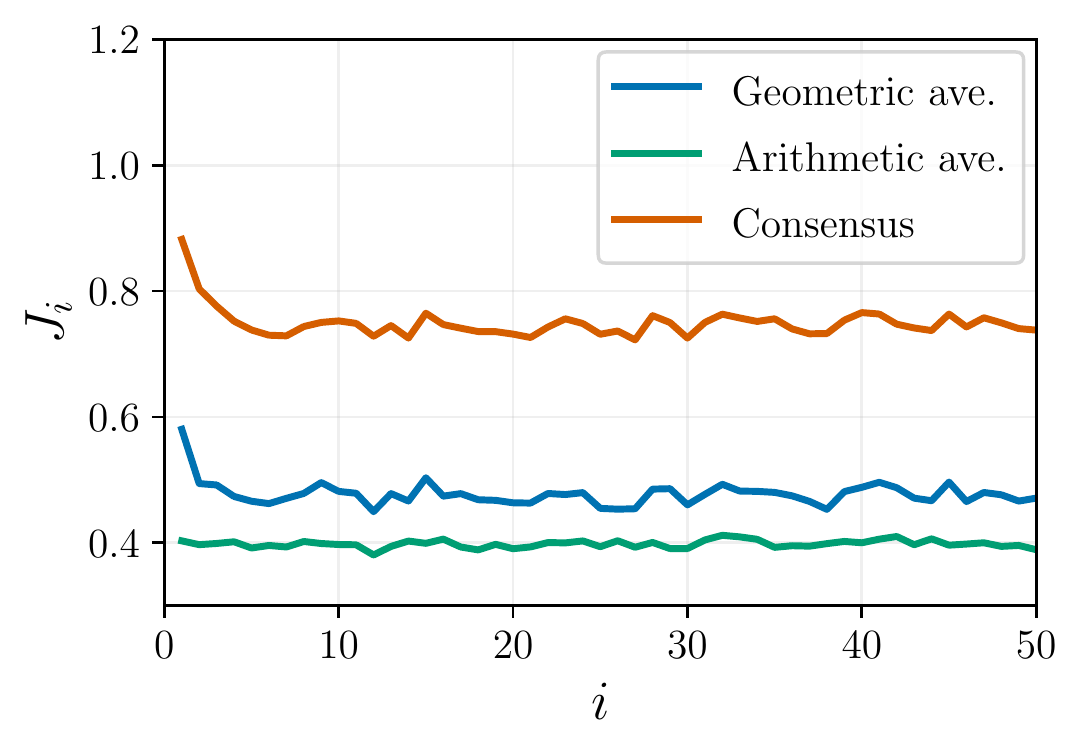}}
  \caption{(a) Average asymptotic risk function and the mixing rate of the network, (b) Average asymptotic time-adjusted prior risk and the Dobrushin coefficient of the transition model, (c) Average risk over time for Diffusion-GA, Diffusion-AA and Consensus-GA.}
  \label{fig:risk} 
\end{figure*}

Fig.~\ref{fig:risk_rho2} illustrates the network average \( J_i\) of asymptotic agent-specific risks $J_{k,i}$ over different network topologies. The sparse networks are associated with higher \( \rho_2 \) values, whereas smaller values arise in dense networks. The risks were approximated by averaging 2000 Monte Carlo simulations with \( \gamma = K \). It can be seen that the average risk is increasing with increasing \( \rho_2\). In other words, the average deviation from the centralized solution decreases with increasing network connectivity. This observation supports Theorem \ref{th:kl_without_net_assumption}. Specifically, when the network is fully-connected, the risk vanishes as expected since the filter is stable (i.e., corrects wrong initialization) as argued in Section \ref{sec:optimality_gap}.

\newcommand{\ra}[1]{\renewcommand{\arraystretch}{#1}}
\newcolumntype{C}{ >{\centering\arraybackslash} m{2cm} }
\newcolumntype{D}{ >{\centering\arraybackslash} m{1.75cm} }
\begin{table}[ht]
\caption{Number of agents and average asymptotic risks across agents \( \frac{1}{K}\sum_{k=1}^K J_{k,\infty} \)}
\vspace{1em}
\small
\ra{1.5}
\begin{tabular}{  D D D D}
  \( K \) & dHMM & DBF \cite{bandyopadhyay_2018} &\( \rho_2 \)\\
 \cline{1-4}
 10   & 0.49& 2.59 &0.86 \\
 
 20    & 0.53&  5.64 &0.82\\
 
  30   & 0.67 &  8.63 &0.81\\
  
  40   & 0.98& 11.88 &0.80 \\
 
 70    & 1.23& 21.29  &0.77\\
 
\end{tabular}
\label{tab:num_agents_risks}
\end{table}

In Table \ref{tab:num_agents_risks}, we compare the average asymptotic risks of networks with different sizes. From the bottom panel of Fig.~\ref{fig:one_realization} we know that choosing \( \gamma \to K \) boosts performance, so we set \( \gamma = K \) for all cases. From Fig.~\ref{fig:risk_rho2} we observe that increasing the network connectivity, i.e., \( \rho_2 \to 0 \), boosts performance. Hence, for a fair comparison, we choose smaller \( \rho_2 \) for larger networks---it is a challenging task to get different-sized graphs with exactly the same \( \rho_2 \). Despite this advantage, larger networks have higher risk values, in other words, higher disagreement with the optimal centralized solution, supporting Theorem \ref{th:kl_without_net_assumption}. We also provide the average risk values for the DBF strategy \cite{bandyopadhyay_2018}. The risk values are significantly higher compared to the proposed dHMM algorithm. Moreover, dHMM is more scalable in the sense that the growth of the risk values with network size is worse in the DBF case.

The effect of \( \kappa (\bT) \) on the average time-adjusted prior divergence \( \widetilde{J}_{\infty} \triangleq  \frac{1}{K}\sum_{k=1}^K \widetilde{J}_{k,\infty} \) can be examined from Fig.~\ref{fig:risk_kappa}. Remember that \( \kappa (\bT)\) is closer to 0 for rapidly mixing transition models. Theorem \ref{th:kl_without_net_assumption} suggests that the risks should increase with increasing \( \kappa (\bT)\). It is visible that this is the case for \( \kappa (\bT) \leq 0.8\), and even more, the risk is equal to 0 for \( \kappa (\bT) = 0\) as revealed by Theorem \ref{th:kl_without_net_assumption}. However, when the informativeness of the observations starts to dominate the ergodicity of the transition model, i.e., \( \kappa (\bT)~\to~1\), the setting gets closer to traditional social learning setup and the risk vanishes, which is unfortunately not explainable with the analysis of the present work. Also note that binary symmetric channels (BSCs) with the same \( \kappa (\bT) \) result in the same divergence. For example, \( \kappa (\bT) = 0.8 \) represents both BSC with change probability \( 0.1 \) and change probability \( 0.9 \). 

In Fig. \ref{fig:risk_comb}, we compare the average risk values of the analyzed diffusion with geometric averaging (GA) to (i) consensus with GA and (ii) diffusion with arithmetic averaging (AA). The age of the utilized information is critical for highly dynamic state transitions. Since diffusion-based strategies use the neighbors' updated information, they outperform the consensus strategy, which can be seen from  Fig. \ref{fig:risk_comb}. Also, diffusion-AA has smaller deviation from the optimal solution compared to the GA-based strategy. However, this observation is not directly transferable to probability of error comparison, as we discuss in the sequel. 


For simulations on probability of error, we consider Gaussian likelihoods, as in Section \ref{sec:gaussian_lklhd}:
\[ L_k(\xi|\theta)  = \begin{dcases} 
       \frac{1}{\sqrt{2\pi}}\exp\left\{-\frac{(\xi+1)^2}{2}\right\},&\theta=0 \\
      \frac{1}{\sqrt{2\pi}}\exp\left\{-\frac{(\xi-1)^2}{2}\right\}, & \theta=1
   \end{dcases}
\]
\begin{figure}[ht]
	\centering
	\includegraphics[width=.8\linewidth]{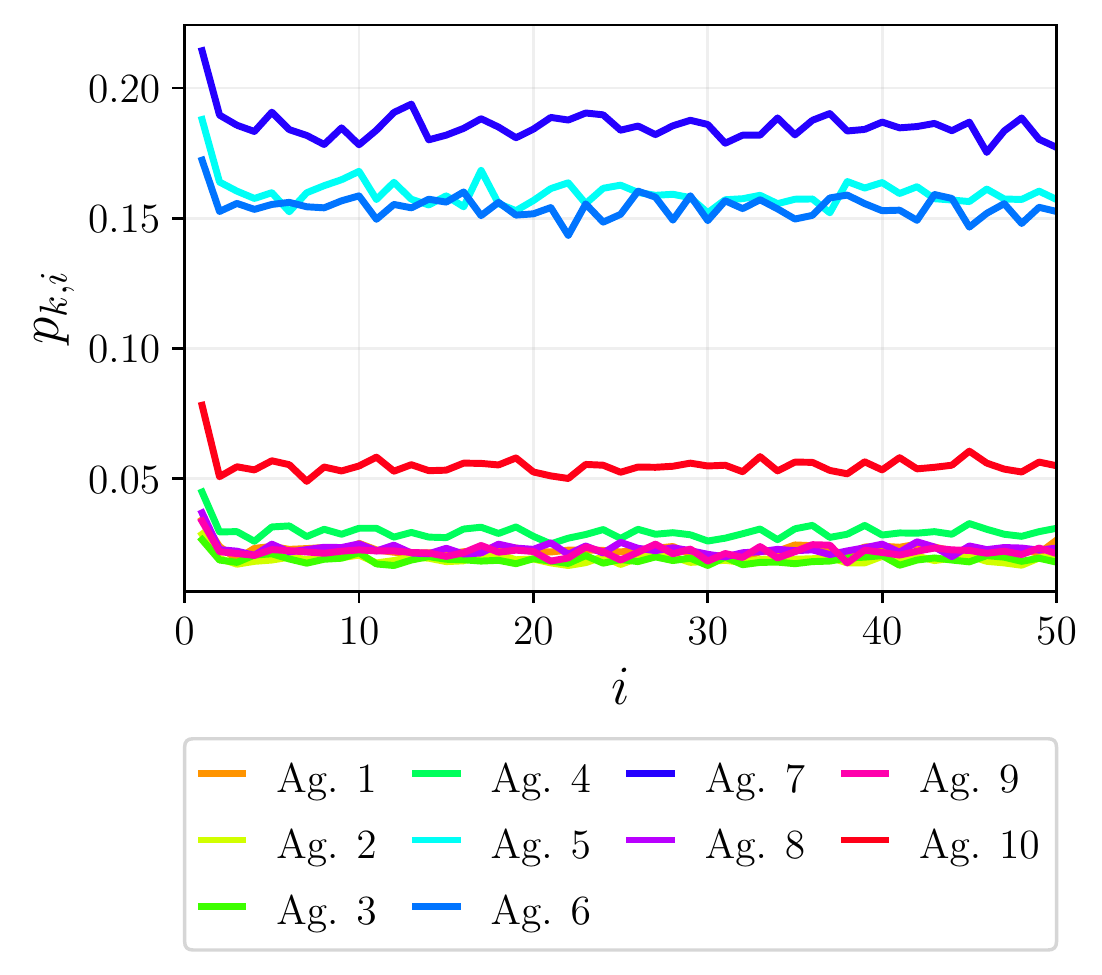}
	\caption{Probability of error across different agents over time.}
	\label{fig:different_agents_pe}
\end{figure}

The plots for error probability are based on 10000 Monte Carlo simulations. We first see in Fig. \ref{fig:different_agents_pe} that the error probabilities of agents rapidly converge, supporting Theorem \ref{theorem:asymptotic_probability}. Moreover, more central agents are better, i.e., wise, in tracking the state of nature compared to less central agents.

\begin{figure*} 
    \centering
  \subfloat[\label{fig:asl_dhmm_chmm}]{%
       \includegraphics[width=0.28\linewidth]{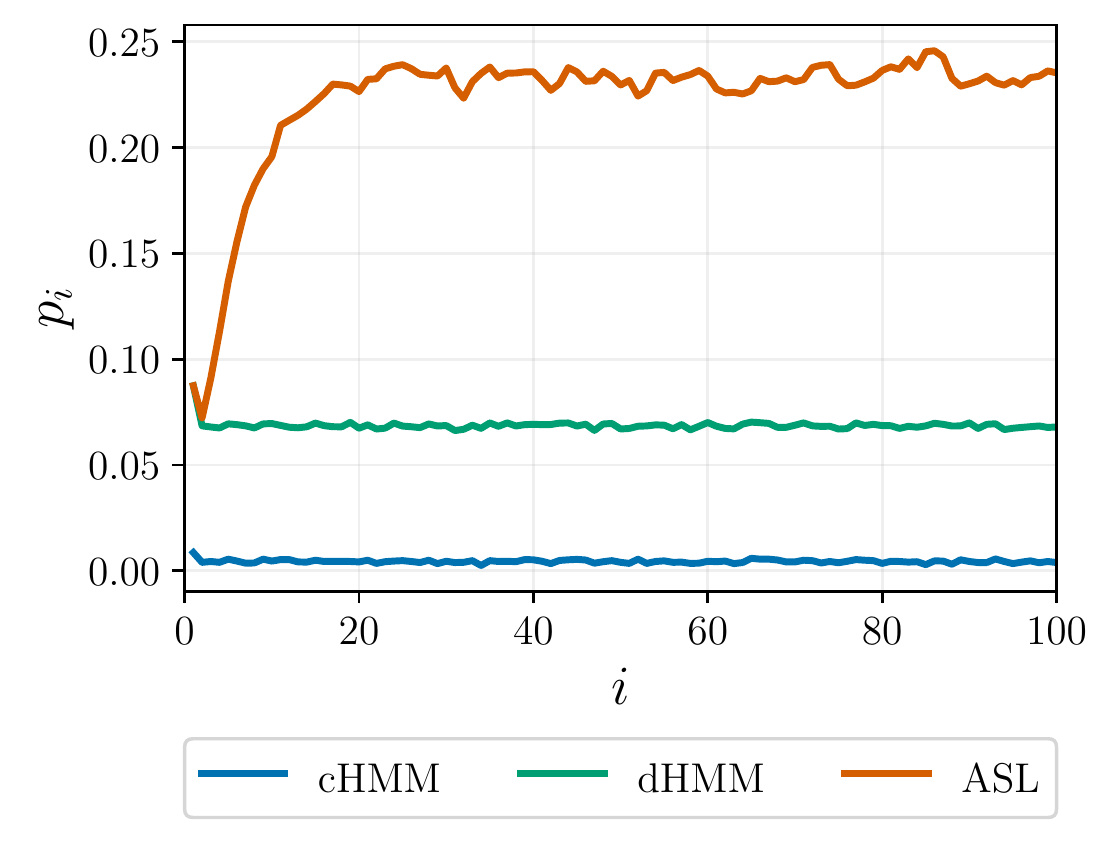}}
    \hfill
  \subfloat[\label{fig:diff_comb}]{%
        \includegraphics[width=0.31\linewidth]{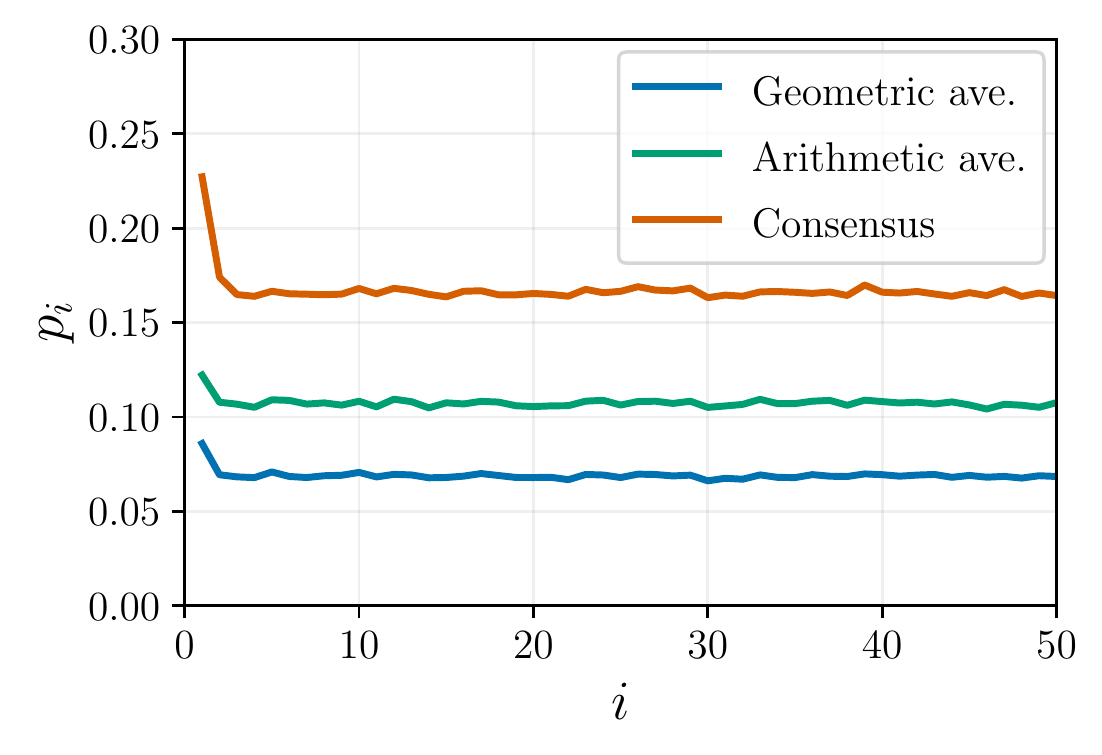}}
        \hfill
  \subfloat[\label{fig:diff_rho2}]{%
       \includegraphics[width=0.31\linewidth]{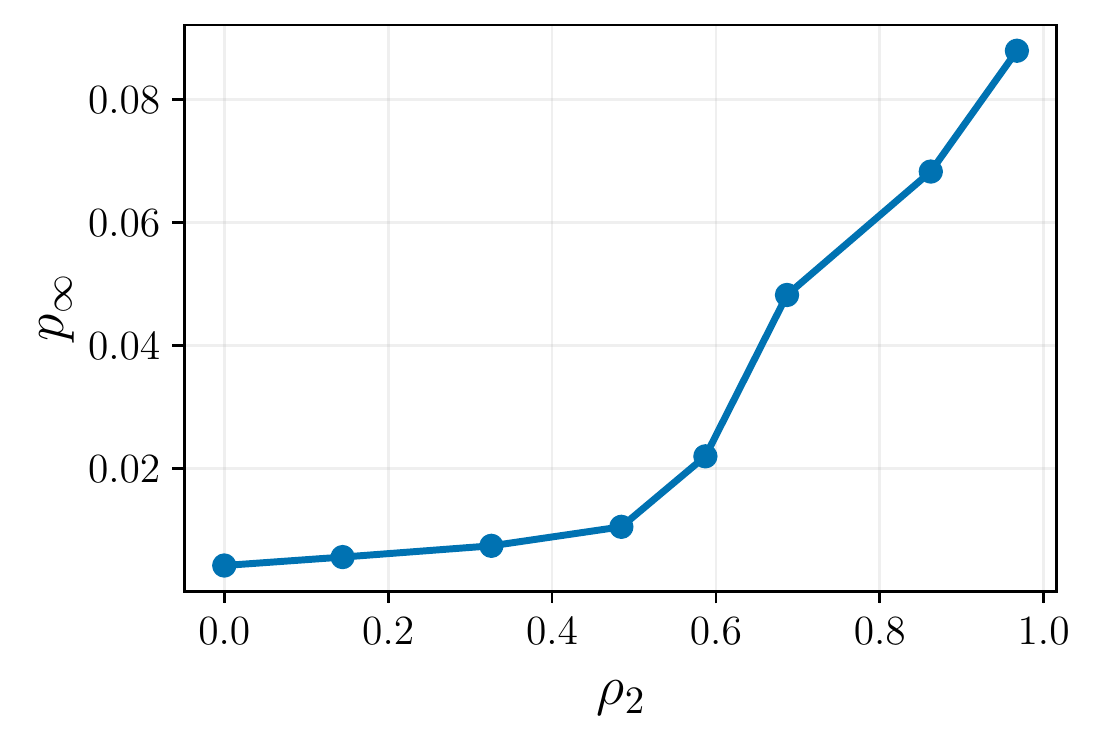}}
  \caption{ Average probability of error over network (a) for cHMM, dHMM and ASL, (b) for diffusion-GA, diffusion-AA, and consensus-GA, over time, (c) in steady-state with respect to different network connectivity.}
  \label{fig:pe} 
\end{figure*}

The network average error probability of diffusion-GA is compared to the (i) centralized and ASL strategies in Fig. \ref{fig:asl_dhmm_chmm} and to (ii) diffusion-AA and consensus-GA in Fig. \ref{fig:diff_comb}. It is seen that the diffusion-GA strategy \eqref{eq:dif_evolve_step}--\eqref{eq:dif_combine_step}  outperforms other distributed solutions, but there is still a gap to the centralized solution which can be removed completely only with fully-connected networks. In particular, diffusion-GA has smaller error probability than diffusion-AA, as opposed to the risk function case. A detailed comparison between these two algorithms can be an interesting future work. Finally, Fig. \ref{fig:diff_rho2} shows that the error probability decreases with increasing network connectivity which highlights the benefit of cooperation.

\section{Concluding Remarks}

In this work, we proposed a distributed and online state estimation algorithm for finite-state HMMs. Based on ergodicity of the underlying transition model, we provided asymptotic bounds on the disagreement between the distributed strategy and the optimal centralized strategy. We also examined the error probability in steady-state and established convergence in distribution under Gaussian observation models.

In addition to state estimation, the proposed algorithm can be used for the prediction of incoming data, by averaging the data-state likelihood functions with respect to the belief over states. More formally, for well-defined cases, agent \( k \) can estimate the incoming data at time \( i \) by using
\begin{align}
\widehat{\bxi}_{k,i} &\triangleq \argmax_{\xi_{k,i}} \sum_{\theta_{i}} L_k(\xi_{k,i}|\theta_i) \bmeta_{k,i}(\theta_i).
\end{align}
Furthermore, the algorithm can also be used for continuous state estimation. This would require changing the summations over the states to integrals. These can be numerically tractable under some conditions. For instance, the exponential family of observation models can lead to closed-form formulas for the integral expressions, as in \cite{dedecius_2017}. 

Finally, the theoretical analysis in the current work utilizes the ergodicity of the transition model to establish performance bounds. A future challenge is to incorporate the informativeness of the observations as well.

\appendices
\allowdisplaybreaks
\section{Proof of Theorem \ref{th:kl_without_net_assumption}}\label{appendix:risk_theorem}
The risk function can be written as
\begin{align}
J_{k,i}&=\e_{\f_i} \dkl(\bmu_i^\star || \bmu_{k,i}) \notag\\
&= \e_{\f_i} \Big [\sum_{\theta_i \in \Theta} \bmu_i^\star (\theta_i) \log \frac{\bmu_i^\star (\theta_i)}{\bmu_{k,i} (\theta_i)} \Big] \notag \\
&\stackrel{(a)}{=} \e_{\f_i} \Big [\sum_{\theta_i \in \Theta} \mathbb{P} (\btheta_i^\circ = \theta_i | \bmf_i) \log \frac{\bmu_i^\star (\theta_i)}{\bmu_{k,i} (\theta_i)} \Big] \notag \\
&\stackrel{(b)}{=}\e_{\f_i} \Big [ \e_{\theta_i^\circ | \f_i} \Big (\log \frac{\bmu_i^\star (\btheta_i^\circ)}{\bmu_{k,i} (\btheta_i^\circ)} \Big ) \Big ] \notag \\
&=\e_{\f_i,\theta_i^\circ}\Big [\log \frac{\bmu_i^\star (\btheta_i^\circ)}{\bmu_{k,i} (\btheta_i^\circ)} \Big ] \notag\\
&\stackrel{(c)}{=}\e_{\f_i,\theta_i^\circ}\Big [\log \bmu_i^\star (\btheta_i^\circ)-\sum_{\ell \in \mathcal{N}_k} a_{\ell k} \log \bpsi_{\ell,i} (\btheta_i^\circ) \Big ] \notag \\ 
&\qquad + \e_{\f_i} \Big [ \log \sum_{\theta_{i}^\prime \in \Theta} \text{exp}\Big \{{\sum_{\ell \in \mathcal{N}_k} a_{\ell k} \log \bpsi_{\ell,i} (\theta_{i}^\prime) \Big \}} \Big ] \notag\\
&=\sum_{\ell \in \mathcal{N}_k} a_{\ell k} \e_{\f_i,\theta_i^\circ} \Big[ \log \frac{\bmu_i^\star (\btheta_i^\circ)}{\bpsi_{\ell,i} (\btheta_i^\circ)} \Big] \notag\\
&\qquad + \e_{\f_i} \Big [ \log \sum_{\theta_{i}^\prime \in \Theta} \text{exp}\Big \{{\sum_{\ell \in \mathcal{N}_k} a_{\ell k} \log \bpsi_{\ell,i} (\theta_{i}^\prime) \Big \}} \Big ] \label{eq:4_term}
\end{align}
where \( (a) \) follows from definition \eqref{eq:true_posterior}, \( (b) \) follows from the definition of conditional expectation with respect to \( \theta_i^\circ \) given \( \f_i\), and \( (c) \) follows from the combine step \eqref{eq:dif_combine_step}. From the centralized update \eqref{eq:centralized_posterior} and the adapt step \eqref{eq:dif_adapt_step}, we have:
\begin{align}
\log \frac{\bmu_i^\star (\theta_i)}{\bpsi_{\ell,i} (\theta_i)}&=\log \frac{L(\bxi_i|\theta_i)}{(L_\ell(\bxi_{\ell,i}|\theta_i))^{\gamma}}+\log \frac{\bmeta_i^\star(\theta_i)}{\bmeta_{\ell,i}(\theta_i)} \notag\\ 
&\qquad -\log \frac{\m_i^\star(\bxi_i)}{\m_{\ell,i}(\bxi_{\ell,i})}. \label{eq:belief_int_belief_ratio}
\end{align}
where we are introducing the following marginal distribution for the new data given the past data:
\begin{align}
\m_i^\star(\xi_i) \triangleq \mathbb{P} (\bxi_i=\xi_i | \bmf_{i-1} ) &= \sum_{\theta_{i}^\prime \in \Theta} \mathbb{P} (\bxi_i=\xi_i, \btheta_i^\circ =\theta_{i}^\prime | \bmf_{i-1} ) \notag \\ &= \sum_{\theta_{i}^\prime \in \Theta} L(\xi_i|\theta_i^\prime)   \mathbb{P} (\btheta_i^\circ =\theta_{i}^\prime | \bmf_{i-1} ) \notag \\ &= \sum_{\theta_{i}^\prime \in \Theta} L(\xi_i|\theta_i^\prime) \bmeta_i^\star(\theta_i^\prime),
\end{align}
as well as the agent-specific pseudo-marginal distribution:
\begin{align}\label{eq:pseudo_marginal}
\m_{\ell,i}(\bxi_{\ell,i}) \triangleq \sum_{\theta_{i}^\prime \in \Theta}(L_\ell(\bxi_{\ell,i} | \theta_{i}^\prime))^{\gamma}\bmeta_{\ell,i} (\theta_{i}^\prime).
\end{align}
Note that \( \m_{\ell,i}(\bxi_{\ell,i}) \) is not a real distribution, i.e., it is not summing up to one because \( \gamma \neq 1 \), in general. To rewrite \eqref{eq:4_term} using \eqref{eq:belief_int_belief_ratio}, we first observe that
\begingroup
\allowdisplaybreaks
\begin{align}
\sum_{\ell \in \mathcal{N}_k} a_{\ell k} \e_{\xi_i,\theta_i^\circ} &\Big [\log \frac{L(\bxi_i|\btheta_i^\circ)}{(L_\ell(\bxi_{\ell,i}|\btheta_i^\circ))^{\gamma}} \Big ]\notag \\
& \stackrel{(a)}{=} \e_{\xi_i,\theta_i^\circ} \Big[ \sum_{\ell=1}^K \log L_\ell(\bxi_{\ell,i} | \btheta_{i}^\circ) \Big ] \notag \\ &\qquad- \sum_{\ell \in \mathcal{N}_k} a_{\ell k} \e_{\xi_{\ell,i},\theta_i^\circ} \Big[ \gamma  \log L_\ell(\bxi_{\ell,i} | \btheta_{i}^\circ) \Big]  \notag\\
& =  \e_{\xi_i,\theta_i^\circ} \Big[ \sum_{\ell=1}^K (1  - \gamma a_{\ell k}) \log L_\ell(\bxi_{\ell,i} | \btheta_{i}^\circ)  \Big] \label{eq:first_4_term}
\end{align}
\endgroup
where in \( (a) \) we used the independence from Assumption \ref{as:independence}. Moreover, the divergence of time-adjusted priors can be bounded as:
\begingroup
\allowdisplaybreaks
\begin{align}
&\sum_{\ell \in \mathcal{N}_k} a_{\ell k} \e_{\f_i,\theta_i^\circ} \Big[  \log \frac{\bmeta_i^\star(\btheta_i^\circ)}{\bmeta_{\ell,i}(\btheta_i^\circ)} \Big ] \notag\\
&\qquad = \sum_{\ell \in \mathcal{N}_k} a_{\ell k} \e_{\f_{i-1}, \theta_i^\circ} \Big[ \e_{\xi_i|\f_{i-1}, \theta_i^\circ}  \Big(  \log \frac{\bmeta_i^\star(\btheta_i^\circ)}{\bmeta_{\ell,i}(\btheta_i^\circ)} \Big ) \Big ] \notag\\
&\qquad \stackrel{(a)}{=} \sum_{\ell \in \mathcal{N}_k} a_{\ell k} \e_{\f_{i-1},\theta_i^\circ} \Big[  \log \frac{\bmeta_i^\star(\btheta_i^\circ)}{\bmeta_{\ell,i}(\btheta_i^\circ)} \Big ] \notag\\
&\qquad = \sum_{\ell \in \mathcal{N}_k} a_{\ell k} \e_{\f_{i-1}} \Big[ \e_{\theta_i^\circ | \f_{i-1}}  \Big(  \log \frac{\bmeta_i^\star(\btheta_i^\circ)}{\bmeta_{\ell,i}(\btheta_i^\circ)} \Big ) \Big ] \notag\\
&\qquad = \sum_{\ell \in \mathcal{N}_k} a_{\ell k} \e_{\f_{i-1}} \Big[\sum_{\theta_i \in \Theta} \mathbb{P} (\btheta_i^\circ = \theta_i | \bmf_{i-1}) \log \frac{\bmeta_i^\star(\theta_i)}{\bmeta_{\ell,i}(\theta_i)} \Big ] \notag\\
&\qquad \stackrel{(b)}{=} \sum_{\ell \in \mathcal{N}_k} a_{\ell k} \e_{\f_{i-1}} \Big[\sum_{\theta_i \in \Theta} \bmeta_i^\star(\theta_i) \log \frac{\bmeta_i^\star(\theta_i)}{\bmeta_{\ell,i}(\theta_i)} \Big ] \notag\\
&\qquad = \sum_{\ell \in \mathcal{N}_k} a_{\ell k} \e_{\f_{i-1}} \Big [\dkl (\bmeta_i^\star || \bmeta_{\ell,i} ) \Big ] \notag\\
&\qquad \stackrel{(c)}{\leq} \sum_{\ell \in \mathcal{N}_k} a_{\ell k} \kappa (\bT) \e_{\f_{i-1}} \Big [\dkl (\bmu_{i-1}^\star || \bmu_{\ell,i-1} ) \Big ] \label{eq:second_4_term}
\end{align}
\endgroup
where \( (a) \) follows from the fact that the time-adjusted priors evaluated at the true hypothesis are deterministic given the old history and the true hypothesis, \( (b) \) follows from definition \eqref{eq:ck_cent}, and \( (c) \) follows from the strong data processing inequality.

Combining \eqref{eq:4_term}, \eqref{eq:belief_int_belief_ratio}, \eqref{eq:first_4_term}, and \eqref{eq:second_4_term} yields:
\begin{align}
J_{k,i} &\leq  \e_{\xi_i,\theta_i^\circ} \Big[ \sum_{\ell=1}^K (1  - \gamma a_{\ell k}) \log L_\ell(\bxi_{\ell,i} | \btheta_{i}^\circ)  \Big]  \notag\\
&\qquad + \sum_{\ell \in \mathcal{N}_k} a_{\ell k} \kappa (\bT) \e_{\f_{i-1}} \Big [\dkl (\bmu_{i-1}^\star || \bmu_{\ell,i-1} ) \Big ] \notag\\
&\qquad +   \e_{\f_i} \Big [  \log \sum_{\theta_{i}^\prime \in \Theta} \text{exp}\Big \{{\sum_{\ell \in \mathcal{N}_k} a_{\ell k} \log \bpsi_{\ell,i} (\theta_{i}^\prime) \Big \}} \Big ] \notag\\
&\qquad - \e_{\f_i} \Big [\sum_{\ell \in \mathcal{N}_k} a_{\ell k} \log \frac{\m_i^\star(\bxi_i)}{\m_{\ell,i}(\bxi_{\ell,i})} \Big ]. \label{eq:risk_first_bound}
\end{align}
\begingroup
\allowdisplaybreaks
Furthermore, the normalization term satisfies:
\begin{align}
&\e_{\f_i} \bigg [  \log \sum_{\theta_{i}^\prime \in \Theta} \text{exp} \Big \{{\sum_{\ell \in \mathcal{N}_k} a_{\ell k} \log \bpsi_{\ell,i} (\theta_{i}^\prime) \Big \}} \bigg ] \notag\\
&\qquad=\e_{\f_i} \bigg [  \log \sum_{\theta_{i}^\prime \in \Theta} \prod_{\ell=1}^K \text{exp} \Big \{{ a_{\ell k} \log \bpsi_{\ell,i} (\theta_{i}^\prime) \Big \}} \bigg ] \notag\\
&\qquad = \e_{\f_i} \bigg [  \log \sum_{\theta_{i}^\prime \in \Theta}  \prod_{\ell=1}^K (\bpsi_{\ell,i} (\theta_{i}^\prime))^{a_{\ell k}} \bigg ] \notag\\
&\qquad \stackrel{(a)}{=}  \e_{\f_i} \bigg [ \log \sum_{\theta_{i}^\prime \in \Theta} \Big ( \prod_{\ell=1}^K (L_\ell(\bxi_{\ell,i} | \theta_{i}^\prime))^{\gamma a_{\ell k}} \prod_{\ell=1}^K  (\bmeta_{\ell,i} (\theta_i^\prime))^{a_{\ell k }} \Big ) \bigg] \notag\\ &\qquad \qquad \qquad -  \e_{\f_i} \bigg [ \sum_{\ell \in \mathcal{N}_k} a_{\ell k} \log  \m_{\ell,i}(\bxi_{\ell,i})   \bigg ] \label{eq:normal_term_exp}
\end{align}
\endgroup
where \( (a) \) follows from \eqref{eq:dif_adapt_step} and \eqref{eq:pseudo_marginal}. Therefore, the last two terms in \eqref{eq:risk_first_bound} can be bounded as:
\begingroup
\allowdisplaybreaks
\begin{align}
& \e_{\f_i} \Big [  \log \sum_{\theta_{i}^\prime \in \Theta} \text{exp} \Big \{{\sum_{\ell \in \mathcal{N}_k} a_{\ell k} \log \bpsi_{\ell,i} (\theta_{i}^\prime) \Big \}} \Big ] \notag\\
&\qquad - \e_{\f_i} \Big [\sum_{\ell \in \mathcal{N}_k} a_{\ell k} \log \frac{\m_i^\star(\bxi_i)}{\m_{\ell,i}(\bxi_{\ell,i})} \Big ] \notag\\
&\qquad \stackrel{(a)}{=} \e_{\f_i} \bigg [ \log \sum_{\theta_{i}^\prime \in \Theta} \Big( \prod_{\ell=1}^K (L_\ell(\bxi_{\ell,i} | \theta_{i}^\prime))^{\gamma a_{\ell k}} \prod_{\ell=1}^K  (\bmeta_{\ell,i} (\theta_i^\prime))^{a_{\ell k }} \Big ) \bigg ] \notag\\
&\qquad \qquad  - \e_{\f_i} \bigg [\sum_{\ell \in \mathcal{N}_k} a_{\ell k} \log \m_{\ell,i}(\bxi_{\ell,i})  \bigg ] \notag \\
&\qquad \qquad - \e_{\f_i} \Big [\sum_{\ell \in \mathcal{N}_k} a_{\ell k} \log \frac{\m_i^\star(\bxi_i)}{\m_{\ell,i}(\bxi_{\ell,i})}  \bigg ] \notag \\
&\qquad = \e_{\f_i} \bigg [ \log \sum_{\theta_{i}^\prime \in \Theta} \Big( \prod_{\ell=1}^K (L_\ell(\bxi_{\ell,i} | \theta_{i}^\prime))^{\gamma a_{\ell k}} \prod_{\ell=1}^K  (\bmeta_{\ell,i} (\theta_i^\prime))^{a_{\ell k }} \Big ) \bigg ] \notag\\
&\qquad \qquad - \e_{\f_i} \bigg [ \log  \m_i^\star(\bxi_i)  \bigg ] \notag \\
&\qquad \stackrel{(b)}{\leq} \e_{\f_i} \bigg [ \log \sum_{\theta_{i}^\prime \in \Theta} \Big ( \prod_{\ell=1}^K (L_\ell(\bxi_{\ell,i} | \theta_{i}^\prime))^{\gamma a_{\ell k}} \sum_{\ell=1}^K a_{\ell k} \bmeta_{\ell,i} (\theta_i^\prime) \Big ) \bigg ]\notag \\
&\qquad \qquad - \e_{\f_i} \bigg [ \log  \m_i^\star(\bxi_i)  \bigg ]  \notag\\
&\qquad  = \e_{\f_i} \bigg [ \log \sum_{\theta_{i}^\prime \in \Theta} \Big ( \prod_{\ell=1}^K (L_\ell(\bxi_{\ell,i} | \theta_{i}^\prime))^{\gamma a_{\ell k}} \sum_{\ell=1}^K a_{\ell k} \bmeta_{\ell,i} (\theta_i^\prime)\Big ) \bigg ] \notag\\
&\qquad \quad - \e_{\f_i} \bigg [ \log \sum_{\theta_{i}^\prime \in \Theta} \Big ( \prod_{\ell=1}^K L_\ell(\bxi_{\ell,i} | \theta_{i}^\prime) \sum_{\ell=1}^K a_{\ell k} \bmeta_{\ell,i} (\theta_i^\prime) \Big ) \bigg] \notag\\
&\qquad \quad + \e_{\f_i} \bigg [ \log \sum_{\theta_{i}^\prime \in \Theta} \Big ( \prod_{\ell=1}^K L_\ell(\bxi_{\ell,i} | \theta_{i}^\prime) \sum_{\ell=1}^K a_{\ell k} \bmeta_{\ell,i} (\theta_i^\prime) \Big ) \bigg ] \notag\\
&\qquad \quad - \e_{\f_i} \Big [\log  \m_i^\star(\bxi_i) \Big ]  \notag\\
&\qquad  \stackrel{(c)}{\leq} \e_{\f_i} \bigg [ \log \sum_{\theta_{i}^\prime \in \Theta} \Big ( \prod_{\ell=1}^K (L_\ell(\bxi_{\ell,i} | \theta_{i}^\prime))^{\gamma a_{\ell k}} \sum_{\ell=1}^K a_{\ell k} \bmeta_{\ell,i} (\theta_i^\prime)\Big ) \bigg ] \notag\\
&\qquad \quad - \e_{\f_i} \bigg [ \log \sum_{\theta_{i}^\prime \in \Theta} \Big ( \prod_{\ell=1}^K L_\ell(\bxi_{\ell,i} | \theta_{i}^\prime) \sum_{\ell=1}^K a_{\ell k} \bmeta_{\ell,i} (\theta_i^\prime) \Big ) \bigg ] \label{eq:nus_before_defined}
\end{align}
\endgroup
where \( (a) \) follows from inserting \eqref{eq:normal_term_exp}, \( (b) \) follows from the weighted arithmetic-geometric mean inequality, \( (c) \) follows from the fact that:
\begin{align}
 &- \e_{\f_i} \Bigg [\log \frac{ \m_i^\star(\bxi_i)}{\sum_{\theta_{i}^\prime \in \Theta} \Big [ \prod_{\ell=1}^K (L_\ell(\bxi_{\ell,i} | \theta_{i}^\prime))  \sum_{\ell=1}^K a_{\ell k} \bmeta_{\ell,i} (\theta_i^\prime)\Big ] }  \Bigg ] \notag\\
 &\qquad \qquad=-\e_{\f_{i-1}} \e_{\xi_i| \f_{i-1}} \Bigg [\log \frac{ \m_i^\star(\bxi_i)}{\m_i^\dagger(\bxi_i)}  \Bigg ] \notag\\
 &\qquad \qquad = - \e_{\f_{i-1}} \dkl(\m_i^\star(\bxi_i) || \m_i^\dagger(\bxi_i)) \notag\\
 &\qquad \qquad \leq 0
\end{align}
where we defined the probability density function:
\begin{align}
\m_i^\dagger(\bxi_i) \triangleq \sum_{\theta_{i}^\prime \in \Theta} \Big [ \prod_{\ell=1}^K (L_\ell(\bxi_{\ell,i} | \theta_{i}^\prime))  \sum_{\ell=1}^K a_{\ell k} \bmeta_{\ell,i} (\theta_i^\prime)\Big ] 
\end{align}
which can be verified to be a density as follows:
\begin{align}
    \int_{\xi_i} \m_i^\dagger(\xi_i) d\xi_i &= \int_{\xi_i} \sum_{\theta_{i}^\prime \in \Theta} \Big [ \prod_{\ell=1}^K (L_\ell(\xi_{\ell,i} | \theta_{i}^\prime))  \sum_{\ell=1}^K a_{\ell k} \bmeta_{\ell,i} (\theta_i^\prime)\Big ] d\xi_i \notag \\
    &=  \sum_{\theta_{i}^\prime \in \Theta} \Big [ \underbrace{\int_{\xi_i} \prod_{\ell=1}^K (L_\ell(\xi_{\ell,i} | \theta_{i}^\prime)) d\xi_i}_{1} \sum_{\ell=1}^K a_{\ell k} \bmeta_{\ell,i} (\theta_i^\prime)\Big ]  \notag \\
    &= \sum_{\theta_{i}^\prime \in \Theta} \Big [  \sum_{\ell=1}^K a_{\ell k} \bmeta_{\ell,i} (\theta_i^\prime)\Big ] \notag \\
    &= \sum_{\ell=1}^K a_{\ell k}  \Big [ \sum_{\theta_{i}^\prime \in \Theta}    \bmeta_{\ell,i} (\theta_i^\prime)\Big ] \notag \\
    &=1
\end{align}
Let us introduce the following vectors of dimension \( H \) over all hypotheses for notational convenience:
\begin{align}\label{eq:vi_plus_def}
 \bvart_i^+ \triangleq \mcl \Bigg \{ \log \Big (\prod_{\ell=1}^K (L_\ell(\bxi_{\ell,i} | \theta_i)^{\gamma a_{\ell k}} \sum_{\ell=1}^K a_{\ell k} \bmeta_{\ell,i} (\theta_i) \Big ) \Bigg \}_{\theta_i=0}^{H-1}
\end{align}
and
\begin{align}\label{eq:vi_neg_def}
 \bvart_i^- \triangleq \mcl \Bigg \{ \log \Big ( \prod_{\ell=1}^K L_\ell(\bxi_{\ell,i} | \theta_i) \sum_{\ell=1}^K a_{\ell k} \bmeta_{\ell,i} (\theta_i) \Big ) \Bigg \}_{\theta_i=0}^{H-1}.
\end{align}
Then, the bound in \eqref{eq:nus_before_defined} can be expressed as
\begin{align}\label{eq:vi_difference_ours}
&\e_{\f_i} \Big [ \log \sum_{\theta_{i}^\prime \in \Theta} \Big [ \prod_{\ell=1}^K (L_\ell(\bxi_{\ell,i} | \theta_{i}^\prime))^{\gamma a_{\ell k}} \sum_{\ell=1}^K a_{\ell k} \bmeta_{\ell,i} (\theta_i^\prime)\Big ] \notag\\
&\quad-\e_{\f_i} \Big [ \log \sum_{\theta_{i}^\prime \in \Theta} \Big [ \prod_{\ell=1}^K L_\ell(\bxi_{\ell,i} | \theta_{i}^\prime) \sum_{\ell=1}^K a_{\ell k} \bmeta_{\ell,i} (\theta_i^\prime) \Big ] \notag\\
&=\!\!\e_{\f_i} \Big [ \log \sum_{\theta_{i}^\prime \in \Theta}\!\! \exp \{ \bvart_i^+(\theta_i^\prime) \} \Big ] -\e_{\f_i} \Big [ \log \sum_{\theta_{i}^\prime \in \Theta}\!\! \exp \{ \bvart_i^-(\theta_i^\prime) \} \Big ] .
\end{align}
Note also that the difference of the vectors satisfy
\begin{align}\label{eq:upsilon_difference}
&  \bvart_i^+  -  \bvart_i^- = \mcl \Big \{  \sum_{\ell=1}^K (  \gamma a_{\ell k}-1) \log L_\ell(\bxi_{\ell,i} | \theta_i)  \Big\}_{\theta_i=0}^{H-1}.
\end{align}
It is useful to introduce the LogSumExp function \( f \):
\begin{align}\label{eq:f_logsumexp_def}
    f(\upsilon) \triangleq \log \sum_{\theta \in \Theta} \exp \{\upsilon (\theta) \},
\end{align}
whose gradient is given by
\begin{align}\label{eq:logsumexp_grad_def}
    \nabla_{\upsilon} f ( \upsilon) \triangleq \mcl \Big \{ \frac{\partial f(\upsilon)}{\partial  \upsilon (\theta)} \Big\}_{\theta \in \Theta}=\mcl \Big \{ \frac{\exp \{\upsilon(\theta)\}}{\sum_{\theta^\prime} \exp \{\upsilon(\theta^\prime)\}} \Big \}_{\theta \in \Theta}.
\end{align}
By applying the mean-value theorem (MVT) to function $f$ and taking the expectation we get
\begin{align}
&\e_{\f_i} \Big [ \log \sum_{\theta_{i}^\prime \in \Theta} \exp \{ \bvart_i^+(\theta_i^\prime) \}\Big ] -\e_{\f_i} \Big [ \log \sum_{\theta_{i}^\prime \in \Theta} \exp \{ \bvart_i^-(\theta_i^\prime) \} \Big ] \notag \\
&\quad \stackrel{\eqref{eq:f_logsumexp_def}}{=}\e_{\f_i} \Big [ f( \bvart_i^+) \Big ] -\e_{\f_i} \Big [ f( \bvart_i^-) \Big ] \notag \\
&\quad \stackrel{\text{(MVT)}}{=} \e_{\f_i} \bigg [ (\nabla_{\upsilon} f ( \bvart_i))^{\T}   \cdot (\bvart_i^+ - \bvart_i^-)\bigg ] \notag \\
&\quad \stackrel{\eqref{eq:upsilon_difference},\eqref{eq:logsumexp_grad_def}}{=}\e_{\f_i} \bigg [ \mcl \Big \{ \frac{\exp \{ \bvart_i(\theta_i) \}}{\sum_{\theta_{i}^\prime} \exp \{ \bvart_i(\theta_{i}^\prime) \} } \Big \}^{\T}  \notag \\& \quad \qquad \quad \cdot \mcl \Big \{  \sum_{\ell=1}^K (  \gamma a_{\ell k}-1) \log L_\ell(\bxi_{\ell,i} | \theta_i)  \Big\}_{\theta_i=0}^{H-1} \bigg ] \label{eq:first_mean_val}
\end{align}
for some  \( \bvart_i \) lying on the line segment between \( \bvart_i^-  \) and \( \bvart_i^+  \). The absolute value of \eqref{eq:first_mean_val} can be bounded for any time instant \( i \):
\begin{align}
 \Bigg | \e_{\f_i}& \Bigg [ \mcl \Big \{ \frac{\exp \{ \bvart_i(\theta_i) \}}{\sum_{\theta_{i}^\prime} \exp \{ \bvart_i(\theta_{i}^\prime) \} } \Big \}^{\T} \notag \\ &\qquad \qquad\cdot \mcl \Big \{  \sum_{\ell=1}^K (  \gamma a_{\ell k}-1) \log L_\ell(\bxi_{\ell,i} | \theta_i)  \Big\}_{\theta_i=0}^{H-1} \Bigg ] \Bigg | \notag\\
 &\stackrel{(a)}{\leq}\e_{\f_i} \Bigg | \mcl \Big \{ \frac{\exp \{ \bvart_i(\theta_i) \}}{\sum_{\theta_{i}^\prime} \exp \{ \bvart_i(\theta_i^\prime)\} } \Big \}^{\T}  \notag \\ &\qquad \qquad \cdot \mcl \Big \{  \sum_{\ell=1}^K (  \gamma a_{\ell k}-1) \log L_\ell(\bxi_{\ell,i} | \theta_i)  \Big\}_{\theta_i=0}^{H-1} \Bigg | \notag \\
 &\stackrel{(b)}{\leq} \e_{\xi_i} \Bigg \| \mcl \Big \{  \sum_{\ell=1}^K (  \gamma a_{\ell k}-1) \log L_\ell(\bxi_{\ell,i} | \theta_i)  \Big\}_{\theta_i=0}^{H-1}  \Bigg \|_{\infty} \notag\\
&\stackrel{(c)}{\leq} \sqrt{K} \gamma \lambda \e_{\xi_i} \Big \| \mcl \Big \{\log L_\ell(\bxi_{\ell,i} | \theta_i)  \Big\}_{\theta_i=0}^{H-1} \Big \|_{\infty} \notag \\
&\stackrel{(d)}{\leq} \sqrt{K} \gamma \lambda C_L \label{eq:normalization_term_bound_small}
\end{align}
where \( (a) \) follows from Jensen's inequality, \( (b) \) follows from Hölder's inequality and the fact that
\begin{align}
    \Big \| \mcl \Big \{ \frac{\exp \{ \bvart_i(\theta_i) \}}{\sum_{\theta_{i}^\prime} \exp \{ \bvart_i(\theta_{i}^\prime) \} } \Big \} \Big \|_1 = 1,
\end{align}
the last step \( (d) \) follows from Assumption~\ref{as:likelihood_functions}, and the step \( (c) \) follows from
\begin{align}
\sum_{\ell=1}^K  |  \gamma a_{\ell k}-1 | &\leq  \gamma   \Big \|A - \frac{1}{\gamma} \mathds{1}_{K}\mathds{1}_{K}^{\T} \Big  \|_{1}  \notag \\
&\leq  \gamma \sqrt{K}  \Big \|A - \frac{1}{\gamma} \mathds{1}_{K}\mathds{1}_{K}^{\T} \Big  \|_{2} \notag \\
&\stackrel{(e)}{=} \gamma \sqrt{K} \lambda. \label{eq:abs_val_of_graph_sum}
\end{align}
Step \( (e) \) follows from the fact that for symmetric matrices, their \( \ell_2 \)-induced norm is equal to the spectral radius. Here, we also use the fact that since \( A \) is primitive and doubly-stochastic, it has a unique eigenvalue at 1, and all other eigenvalues lie inside the unit circle. Moreover, it is simultaneously diagonalizable (i.e., they have the same eigenvectors) with the all ones matrix \( \mathds{1}_{K}\mathds{1}_{K}^{\T} \) \cite[Chapter 8]{horn2012matrix}. As a result, the corresponding spectral radius becomes the maximum absolute difference between the eigenvalues of \( A \) and  \( \frac{1}{\gamma} \mathds{1}_{K}\mathds{1}_{K}^{\T} \), i.e., \(\lambda\). Next, we combine \eqref{eq:risk_first_bound}, \eqref{eq:nus_before_defined}, and \eqref{eq:normalization_term_bound_small} to obtain the bound on the risk function:
\begin{align}
J_{k,i}  &\leq \sum_{\ell=1}^K (1  - \gamma a_{\ell k}) \e_{\xi_{\ell,i},\theta_i^\circ} \Big[  \log L_\ell(\bxi_{\ell,i} | \btheta_{i}^\circ)  \Big]  \notag\\
&\quad + \kappa (\bT) \sum_{\ell \in \mathcal{N}_k} a_{\ell k}  \underbrace{\e_{\f_{i-1}} \Big [\dkl (\bmu_{i-1}^\star || \bmu_{\ell,i-1} ) \Big ]}_{J_{\ell,i-1}} \notag\\
&\qquad + \sqrt{K} \gamma \lambda C_L. \label{eq:combine_risk_mean}
\end{align}
Iterating this bound over time for all agents results in
\begin{align}
J_{k,i} &\leq \sum_{\ell=1}^K (1  - \gamma a_{\ell k}) \e_{\xi_{\ell,i},\theta_i^\circ} \Big[  \log L_\ell(\bxi_{\ell,i} | \btheta_{i}^\circ)  \Big]  \notag\\
&\quad + \kappa (\bT) \sum_{\ell \in \mathcal{N}_k} a_{\ell k} \sum_{m \in \mathcal{N}_\ell} \Big ( (1  - \gamma a_{m \ell}) \notag \\ &\qquad \qquad \quad  \times \e_{\xi_{m,i-1},\theta_{i-1}^\circ} \Big[  \log L_m(\bxi_{m,i-1} | \btheta_{i-1}^\circ)  \Big] \Big ) \notag\\
&\quad + \kappa (\bT) \sum_{\ell \in \mathcal{N}_k} a_{\ell k} \kappa (\bT) \sum_{m \in \mathcal{N}_\ell} a_{m \ell} J_{m,i-2} \notag\\
&\quad + \kappa (\bT) \sum_{\ell \in \mathcal{N}_k} a_{\ell k}  \sqrt{K} \gamma \lambda C_L \notag\\
&\qquad + \sqrt{K} \gamma \lambda C_L \notag \\
&\stackrel{(a)}{\leq} \sum_{\ell=1}^K (1  - \gamma a_{\ell k}) \e_{\xi_{\ell,i},\theta_i^\circ} \Big[  \log L_\ell(\bxi_{\ell,i} | \btheta_{i}^\circ)  \Big]  \notag\\
&\qquad \quad + \kappa (\bT) \sum_{m=1}^K \Big ( (1  - \gamma [A^{2}]_{m k}) \notag \\ &\qquad \qquad \quad  \times \e_{\xi_{m,i-1},\theta_{i-1}^\circ} \Big[  \log L_m(\bxi_{m,i-1} | \btheta_{i-1}^\circ)  \Big] \Big ) \notag\\
&\qquad \quad  + \kappa (\bT)^2  \sum_{m=1}^K [A^{2}]_{m k} J_{m,i-2} \notag\\
&\qquad \quad  + (1+ \kappa (\bT) ) \sqrt{K} \gamma \lambda C_L \notag \\
&\leq \sum_{j=0}^{i-1} (\kappa (\bT))^{j}  \sum_{\ell=1}^K (1-  \gamma [A^{j+1}]_{\ell k})\notag \\&\qquad \qquad \qquad \quad \cdot \e_{\xi_{\ell,i-j},\theta_{i-j}^\circ} \Big[  \log L_\ell(\bxi_{\ell,i-j} | \btheta_{i-j}^\circ)  \Big] \notag\\
& \qquad \quad+ \sum_{j=0}^{i-1} (\kappa (\bT))^{j} \sqrt{K} \gamma \lambda C_L \notag \\
& \qquad \qquad + (\kappa (\bT))^{i} \sum_{\ell =1}^K [A^{i}]_{\ell k} J_{\ell,0} \label{eq:bound_recursion_temp},
\end{align}
where \( (a) \) follows from the fact that
\begin{align}
    &\sum_{\ell \in \mathcal{N}_k}\!\!\! a_{\ell k} \!\!\!\sum_{m \in \mathcal{N}_\ell}\!\!\! \Big ( (1  - \gamma a_{m \ell})  \e_{\xi_{m,i-1},\theta_{i-1}^\circ} \Big[  \log L_m(\bxi_{m,i-1} | \btheta_{i-1}^\circ)  \Big] \Big ) \notag \\
    &=\!\!\!\sum_{m=1 }^K\e_{\xi_{m,i-1},\theta_{i-1}^\circ} \Big[  \log L_m(\bxi_{m,i-1} | \btheta_{i-1}^\circ)  \Big]  \!\sum_{\ell=1}^K \! a_{\ell k}   (1  - \gamma a_{m \ell})  \notag \\
    &=\!\!\!\sum_{m=1 }^K(1  - \gamma [A^{2}]_{m k})\e_{\xi_{m,i-1},\theta_{i-1}^\circ} \Big[  \log L_m(\bxi_{m,i-1} | \btheta_{i-1}^\circ)  \Big].  
\end{align}
The first summation in the bound \eqref{eq:bound_recursion_temp} can be further bounded by the inequality
\begin{align}
\Bigg | &\sum_{\ell = 1}^K  (1-  \gamma [A^{j+1}]_{\ell k} ) \e_{\xi_{\ell,i-j},\theta_{i-j}^\circ} \Big [  \log L_\ell (\bxi_{\ell,i-j} | \btheta_{i-j}^\circ)  \Big ] \Bigg | \notag\\
&\stackrel{(a)}{\leq} \sum_{\ell = 1}^K \Big | 1-  \gamma [A^{j+1}]_{\ell k} \Big | \times \Big |\e_{\xi_{\ell,i-j},\theta_{i-j}^\circ} \Big [  \log L_\ell (\bxi_{\ell,i-j} | \btheta_{i-j}^\circ)  \Big ] \Big | \notag\\
&\stackrel{(b)}{\leq} \sum_{\ell = 1}^K \Big | 1-  \gamma [A^{j+1}]_{\ell k} \Big | C_L  \notag\\
&\stackrel{(c)}{\leq} \sqrt{K} \gamma \lambda_j C_L \label{eq:abs_val_of_graph_sum_w_norm}
\end{align}
where \( \lambda_j > \max \{ |1-\frac{K}{\gamma}|, \rho_2^{j+1} \} \)  is a positive constant, \( (a) \) follows from Jensen's inequality, \( (b) \) follows from Assumption~\ref{as:likelihood_functions}, and \( (c) \) follows from \eqref{eq:abs_val_of_graph_sum} applied to \(A^{j+1}\) instead of \( A \). Inserting \eqref{eq:abs_val_of_graph_sum_w_norm} into \eqref{eq:bound_recursion_temp} we can bound the risk function as:
\begin{align}
J_{k,i} &\leq \sum_{j=0}^{i-1} (\kappa (\bT))^{j}   \sqrt{K} \gamma  \lambda_j C_L  + \sum_{j=0}^{i-1} (\kappa (\bT))^{j} \sqrt{K} \gamma  \lambda C_L \notag\\
&\quad  + (\kappa (\bT))^{i} \sum_{\ell =1}^K [A^{i}]_{\ell k} J_{\ell,0} \notag \\
&\stackrel{(a)}{\leq}  2  \sum_{j=0}^{i-1} (\kappa (\bT))^{j} \sqrt{K} \gamma \lambda C_L + (\kappa (\bT))^{i} \sum_{\ell =1}^K [A^{i}]_{\ell k} J_{\ell,0} \notag \\
&\leq   2   \sum_{j=0}^{i-1} (\kappa (\bT))^{j} \sqrt{K}  \gamma \lambda C_L+ (\kappa (\bT))^{i} \sum_{\ell =1}^K [A^{i}]_{\ell k} J_{\ell,0} \notag \\
&=  2  \frac{1-(\kappa (\bT))^{i}}{1-\kappa (\bT)}  \sqrt{K} \gamma  \lambda C_L + (\kappa (\bT))^{i} \sum_{\ell =1}^K [A^{i}]_{\ell k} J_{\ell,0} 
\end{align}
where \( (a) \) follows from \( \lambda_j \leq \lambda \) for all \( j \). For \( \kappa (\bT) < 1 \), the risk function is asymptotically bounded:
\begin{align}
\limsup_{i \rightarrow \infty} J_{k,i} \leq   \frac{2 \sqrt{K} \gamma  \lambda C_L}{1-\kappa (\bT)}  .
\end{align}
This also means that
\begin{align}
     \widetilde{J}_{k,i} &\stackrel{(a)}{\leq}  \kappa (\bT) J_{k,i-1} \notag \\
    \Longrightarrow \limsup_{i \rightarrow \infty} \widetilde{J}_{k,i} &\leq \limsup_{i \rightarrow \infty} \kappa (\bT) J_{k,i-1} \notag \\
      &\leq \kappa (\bT) \frac{2 \sqrt{K} \gamma  \lambda C_L}{1-\kappa (\bT)}  
\end{align}
where \( (a) \) follows from the strong-data processing inequality, for any time instant \( i \).

\section{An Auxiliary Lemma}\label{appendix:matrix}



We present a general matrix result in the following lemma.

\begin{lemma}[{\bf Lower dimensional representation}]\label{lemma:matrix}
Consider the \(K\times K\) doubly stochastic and symmetric combination matrix \(A\). Let \(r=\rank(A)\). Then, for any positive-definite diagonal covariance matrix \(\Sigma\), there  exists an \( r\times K \) matrix \( Q \) such that:
\begin{itemize}
    \item \( A^{\T} \Sigma A = Q^{\T} Q \),
    \item for any vector \( v \in \mathbb{R}^K \), there exists a unique vector \( v_{Q} \in \mathbb{R}^r \) that satisfies:
    \begin{align}
        A^{\T} v = Q^{\T} v_{Q}.
    \end{align}
    In other words, \( Q \) has full row rank and 
    \begin{align}
        v_{Q} = (Q^{\T} )^\dagger A^{\T} v,
    \end{align}
    where \( (Q^{\T} )^\dagger \) is the pseudo-inverse matrix 
    \begin{align}
     (Q^{\T} )^\dagger  \triangleq (Q Q^{\T})^{-1} Q.
    \end{align}
\end{itemize}
\end{lemma}
\begin{proof}
Observe that 
\begin{align}
    \rank (A) &\stackrel{(a)}{=} \rank (\Sigma^{1/2} A) \notag\\
    &= \rank \big ((\Sigma^{1/2} A)^{\T} \Sigma^{1/2} A \big) \notag \\
    &=\rank (A^{\T} \Sigma A) = r
\end{align}
where \( (a) \) follows from the fact that \( \Sigma \) is positive-definite and $\Sigma^{1/2}$ is its square-root. Moreover, since \( A^{\T} \Sigma A \) is a real symmetric matrix, it can be decomposed as
\begin{align}\label{eq:ata_decomposition}
    A^{\T} \Sigma A = U \Lambda U^{\T}
\end{align}
where \( U \) is \( K\times r\) with orthonormal columns and \( \Lambda \) is \( r\times r \) with the positive eigenvalues of \( A^{\T} \Sigma A \).  Let \( Q=\Lambda^{1/2} U^{\T} \), which has full row rank. Then, \( A^{\T}\Sigma A = Q^{\T}Q \). Note that \( Q \) is not unique since we can modify it by any orthonormal transformation. It is also obvious that, in terms of null ($\mynull$) and range ($\range$) spaces,  
\begin{align}
\mynull(A)&=\mynull (\Sigma^{1/2} A) \notag \\
&=\mynull(A^{\T}\Sigma A) \notag \\
&=\mynull(Q^{\T} Q) \notag \\
&=\mynull(Q)
\end{align}
and, hence, \( \range(A^{T})=\range(Q^{\T})\). It follows that for any vector \( v \in \mathbb{R}^K \), there exists a vector \( v_{Q} \in \mathbb{R}^r \) such that
    \begin{align}\label{eq:a_q_equivalence}
        A^{\T} v = Q^{\T} v_{Q}.
    \end{align}
Multiplying both sides of \eqref{eq:a_q_equivalence} from the left by the pseudo-inverse of $Q^{\T}$ gives
\begin{align}
    v_{Q} = (Q^{\T} )^\dagger A^{\T}  v.
\end{align}
\end{proof}

\section{Error Recursion for Diffusion}\label{appendix:diffusion_recursion}

In light of Lemma \ref{lemma:matrix} from Appendix~\ref{appendix:matrix}, there exist vectors in \( \mathbb{R}^{r}\) such that:
\begin{align}\label{eq:w_wq_transform}
    \w_{i} &=  A^{\T} ( \bm{\nu}_{i} +  \bm{\chi}_{i} ) \notag \\
    &=Q^{\T} (\bm{\nu}_{Q,i} +  \bm{\chi}_{Q,i}) \notag \\
    &=Q^{\T} \w_{Q,i}
\end{align}
where
\begin{align}
    \w_{Q,i} &\triangleq (Q^{\T} )^\dagger  \w_{i} \label{eq:wqi_def}\\
    \bm{\chi}_{Q,i}&\triangleq (Q^{\T} )^\dagger A^{\T} \bm{\chi}_{i} \label{eq:chiq_def} \\
    \bm{\nu}_{Q,i} &\triangleq (Q^{\T} )^\dagger A^{\T} \bm{\nu}_{i} \label{eq:vqi_def}.
\end{align}
Then, it follows from \eqref{eq:v_gaussian} and \eqref{eq:vqi_def} that 
\begin{align}\label{eq:nu_q_i_gaussian}
   \bm{\nu}_{Q,i} \sim \mathcal{G} \Big (  (Q^{\T} )^\dagger A^{\T} \bm{\beta}^{(\theta_{i}^\circ)} , I_{r\times r} \Big ).
\end{align}
where the covariance term follows from Lemma~\ref{lemma:matrix}:
\begin{equation}
 (Q^{\T} )^\dagger A^{\T} \Sigma A Q^\dagger \stackrel{}{=} (Q^{\T} )^\dagger( Q^{\T} Q) Q^\dagger = I_{r\times r}.
\end{equation}
Moreover, from the definition \eqref{eq:chi_def} of \( \bm{\chi}_{i} \) in terms of \( \w_{i-1} \) and \eqref{eq:w_wq_transform} we can alternatively write
\begin{align}\label{eq:xi_wqi1_relation}
    \bm{\chi}_{i} = \text{col} \Bigg \{  \log \frac{\bT(1|0)+\bT(1|1) \text{exp}\{[Q^{\T}\w_{Q,i-1}]_{\ell}\}}{\bT(0|0)+\bT(0|1)\text{exp}\{[Q^{\T}\w_{Q,i-1}]_{\ell}\} } \Bigg \}_{\ell=1}^K.
\end{align}
Now note from \eqref{eq:wqi_def}--\eqref{eq:vqi_def} that
\begin{align}
    \w_{Q,i} &=  \bm{\nu}_{Q,i} +  \bm{\chi}_{Q,i} \notag \\
    &\stackrel{\eqref{eq:chiq_def}}{=} \bm{\nu}_{Q,i} + (Q^{\T} )^\dagger A^{\T} \bm{\chi}_{i}
\end{align}
which represents a transformation from \(\w_{Q,i-1}  \) to \( \w_{Q,i} \) directly in light of \eqref{eq:xi_wqi1_relation}. This indicates that we can work with \( \w_{Q,i} \) over time and transform it into the original vector \( \w_i \) via \eqref{eq:w_wq_transform} without any loss of information. Repeating arguments similar to \eqref{eq:joint_time_density} and \eqref{eq:si_definition} and replacing \( \w_i \) by \( \w_{Q,i} \) we obtain:
\begin{align}\label{eq:si_q_definition}
    &S_{Q,i}^{(\theta)} (w_Q,w_Q^{\prime})dW_Q \notag \\
    &\qquad \triangleq \mathbb{P}(\w_{Q,i} \in (w_Q,w_Q+dw_Q)|\btheta_{i}^\circ =\theta , \w_{Q,i-1} = w_Q^{\prime}) \notag \\
    &S_{Q,i}^{(\theta)} (w_Q,w_Q^{\prime}) \stackrel{(a)}{=} \frac{1}{(2\pi)^{r/2}} \exp{\Big \{ -\frac{1}{2} \|w_Q - \rho_Q^{(\theta)} (w_Q^{\prime}) \|^2 \Big \}}
\end{align}
where \( (a) \) follows from \eqref{eq:nu_q_i_gaussian} and
\begin{align}
    \rho_Q^{(\theta)}  (w_Q^{\prime}) \triangleq (Q^{\T} )^\dagger A^{\T} \bm{\beta}^{(\theta_{i}^\circ)} +  \bm{\chi}_{Q,i} \Big |_{\btheta_{i}^\circ = \theta, \w_{Q,i-1} = w_Q^{\prime}}.
\end{align}
Observe that the density \eqref{eq:si_q_definition} of \( \w_{Q,i}\) always exists in \( \mathbb{R}^r \), even if \( \w_{i} \) does not admit a density in \( \mathbb{R}^K \). Furthermore, the effective temporal recursion becomes
\begin{align}\label{eq:fqialpha_recursion}
   f_{Q,i}(\theta,w_Q)\!\!=\!\!\sum_{\theta^\prime} \bT(\theta| \theta^\prime)\!\! \Bigg [\!\! \int_{w_Q^{\prime}}  \!\!\!\!\! S_{Q,i}^{(\theta)} (w_Q,w_Q^{\prime}) f_{Q,i-1} (\theta^\prime,w_Q^{\prime})dW_Q^\prime \!\Bigg ]
\end{align}
Using this information along with \eqref{eq:w_wq_transform}, which allows us to recover the agent-specific log-belief ratio \( \w_{k,i} \) from the low-dimensional representation \( \w_{Q,i} \), namely,
\begin{align}
    \w_{k,i} = q_{k}^{\T} \w_{Q,i},
\end{align}
where \( q_{k}\) is the \( k \)th column of \( Q \), we arrive at the following probability of error calculation for the diffusion HMM strategy:
\begin{align}
p_{k,i} =     \underset{q_{k}^{\T} w_{Q} \leq 0}{\int ... \int}  f_{Q,i} (1,w_{Q})dW_{Q} + \underset{q_{k}^{\T} w_{Q} > 0}{\int ... \int} f_{Q,i} (0,w_{Q})dW_{Q}.
\end{align}

\section{Proof of Theorem \ref{theorem:asymptotic_probability}}\label{appendix:asymptotic_probability}
First, observe from \eqref{eq:chi_def} that for a given transition model \( 0 < \bT(\theta |\theta^\prime) < 1\) \( \forall \theta,\theta^\prime \in \Theta \), $\bm{\chi}_{i}$ is bounded in norm. Moreover, the Gaussian mean $\bm{\beta}^{(\theta_{i}^\circ)}$ is also bounded---see \eqref{eq:v_gaussian_mean}. These in turn imply:
\begin{equation}
  \| \rho^{(\theta)} (w^{\prime}) \|_{\Sigma^{-1}} \leq  \widetilde{\rho}
\end{equation}
for some constant $\widetilde{\rho}>0$.
Let us define the spherical region \( \mathcal{R} \triangleq \{w: \| w\|_{\Sigma^{-1}} \leq \widetilde{\rho} \}\). For any vector \( w \) that satisfies  \(  \| w\|_{\Sigma^{-1}}  \geq \widetilde{\rho}  \), the projection to this region is given by \( \widetilde{\rho} \frac{w}{\|w\|_{\Sigma^{-1}}} \), which can be verified by following the same steps for finding a vector's projection into \( \ell_2\)-ball. This implies that for any \( w \) outside the region \( \mathcal{R}\), and for any \( w^{\prime} \):
\begin{align}
    \min_{w_p \in \mathcal{R}} \|w-w_p \|_{\Sigma^{-1}} &= \Big \| w - \widetilde{\rho} \frac{w}{\|w\|_{\Sigma^{-1}} } \Big \|_{\Sigma^{-1}} \notag \\ &\leq \| w-\rho^{(\theta)} (w^{\prime}) \|,
\end{align}
since \( \rho^{(\theta)} (w^{\prime}) \in \mathcal{R}\) for any \( w^{\prime}\). Incorporating this with \eqref{eq:si_consensus_expression}, observe that for consensus we have
\begin{align}
    0 < S_{i}^{(\theta)} (w,w^{\prime}) \leq  \widetilde{S} (w)
\end{align}
for a Lebesgue integrable function $\widetilde{S} (w)$:
\begin{equation}
  \widetilde{S} (w) \triangleq  \begin{dcases}
    \frac{1 }{\sqrt{(2\pi)^{K} \text{det}( \Sigma  ) }} , \qquad \|w \|_{\Sigma^{-1}} < \widetilde{\rho}\\
   \frac{\text{exp} \Big \{ -\frac{1}{2} \Big \| w - \widetilde{\rho} \frac{w}{\|w\|_{\Sigma^{-1}} } \Big \|_{\Sigma^{-1}}^2 \Big \} }{\sqrt{(2\pi)^{K} \text{det}( \Sigma  ) }}, \text{elsewhere}
    \end{dcases}.
\end{equation}
 Therefore, the kernel \( \bT(\theta |\theta^\prime) S_{i}^{(\theta)} (w,w^{\prime})\) of the recursion \eqref{eq:fialpha_recursion} satisfies the conditions required by \cite[Theorem 5.7.4]{lasota1998chaos} and we conclude that: 
\begin{align}
    & \lim_{i \to \infty} \big \| f_{i}  - f_{\infty}\big \|_\textup{TV} = 0. \label{eq:app_eq_stability_noq}
\end{align}
 Consider time independent random variables \( \{ \btheta^\circ_{\infty}, \w_{\infty} \} \) whose joint pdf is given by \( f_{\infty} \).  
The convergence in total variation \eqref{eq:app_eq_stability_noq} implies convergence in distribution (defined in \eqref{eq:def_in_dist_conv}), which means that  \( \{ \btheta_{i}^\circ, \w_{i} \}\) converge to limiting random variables \( \{ \btheta^\circ_{\infty}, \w_{\infty} \}\) in distribution, i.e., \( \{ \btheta_{i}^\circ, \w_{i} \} \stackrel{d}{\rightsquigarrow} \{ \btheta^\circ_{\infty}, \w_{\infty} \} \). As a result, if we define for consensus
\begin{align}
     p_{k,\infty} &=\int_{w_k=- \infty}^{0} f_{k,\infty} (1,w_k)dw_k + \int_{w_k=0}^{\infty} f_{k,\infty} (0,w_k)dw_k,
\end{align}
where
\begin{equation}
     f_{k,\infty} (\theta,w_k) =\int \!\!\cdots \!\!\int \! \!  f_{\infty} (\theta,w)dw_1\cdots dw_{k-1}dw_{k+1}\cdots dw_K,
\end{equation}
we obtain the convergence to the steady-state error probability
\begin{align}
    \lim_{i \to \infty} p_{k,i} = p_{k, \infty}.
\end{align}
Similarly for diffusion.

\section{Proof of Lemma~\ref{prop:network_disagreement}}\label{sec:disagreement_example}
We verify that in general there is no network agreement by providing a counter-example. Consider the following special case of the binary hypothesis testing problem described in Sec. \ref{sec:probability_error}: \\

\begin{figure}[ht]
	\centering
	\includegraphics[width=2.9in]{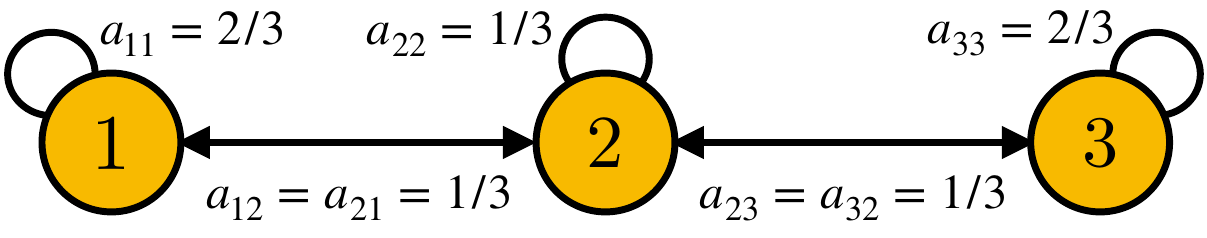}
	\caption{Network of \( K=3\) agents.}
	\label{fig:three_network}
\end{figure}

\noindent\textbf{The model}: Consider the network of \( K=3 \) agents in Fig.~\ref{fig:three_network}. Observe that the network is strongly-connected. Moreover, the combination weights satisfy the doubly-stochastic and symmetric matrix Assumption~\ref{as:network_top}. Assume for simplicity of notation that all agents have the same Gaussian observation models with \( \sigma_k^2=\sigma^2 \). The transition Markov chain is a binary symmetric channel with transition probability \( \alpha = 0.5 \). \\

\noindent\textbf{System equilibrium}: It follows that for each agent \( k \), regardless of the belief \( \bmu_{k,i-1}(\theta_{i-1}) \),
\begin{align}\label{eq:ck_cent_example}
\bmeta_{k,i} (\theta_i) &= \sum_{\theta_{i-1}\in\Theta} \bT(\theta_{i}| \theta_{i-1}) \bmu_{k,i-1}(\theta_{i-1}) \notag \\
&=\alpha \bmu_{k,i-1}(0) + \alpha \bmu_{k,i-1}(1) = 0.5.
\end{align}
This also means that, for each agent \( k \),
\begin{align}
     \log \frac{\bmeta_{k,i}(1)}{\bmeta_{k,i}(0)} = 0,
\end{align}
and according to \eqref{eq:logbelief_agent_rec},
\begin{align}\label{eq:logbelief_agent_rec_example}
    \w_{k,i} = \sum_{\ell \in \mathcal{N}_k} a_{\ell k} &\gamma \log \frac{L_{\ell}(\bxi_{\ell,i} | 1) }{L_{\ell}(\bxi_{\ell,i} | 0)}.
\end{align}
Moreover, since the likelihood functions of the agents are identical, the entries 
\begin{align}
    \bm{\nu}_{k,i} \triangleq \gamma \log \frac{L_{k}(\bxi_{k,i} | 1) }{L_{k}(\bxi_{k,i} | 0)}
\end{align}
of \( \bm{\nu}_{i} \in \mathbb{R}^3 \) are i.i.d. Gaussian random variables (with mean in \eqref{eq:v_gaussian_mean} and variance in \eqref{eq:covar_definition}) given the true hypothesis. For the peripheral agent \( 1 \), it holds that
\begin{align}
     \w_{1,i} = \sum_{\ell \in \mathcal{N}_k} a_{\ell k} \bm{\nu}_{\ell,i} = \frac{2}{3} \bm{\nu}_{1,i} + \frac{1}{3} \bm{\nu}_{2,i},
\end{align}
and for the central agent \( 2 \),
\begin{align}
    \w_{2,i} =  \frac{1}{3} \bm{\nu}_{1,i} + \frac{1}{3} \bm{\nu}_{2,i} + \frac{1}{3} \bm{\nu}_{3,i}.
\end{align}
Observe that the log-belief ratio \(  \w_{k,i} \) for each agent \( k \) is a Gaussian random variable, with mean and variance parameters that do not depend on time. However, even though \( \w_{1,i} \) and \(  \w_{2,i} \) have the same mean, their variances and hence their distributions are not the same. This proves that agents do not converge to the same random variable and do not have the same steady-state error probability. In this particular example, the central agent has less error probability due to smaller variance. For the consensus case, a similar counter-example can be formed by allowing agents to have different likelihoods.

\bibliographystyle{IEEEbib}
\bibliography{refs}






\vfill

\end{document}